\documentclass[11pt]{article}

\usepackage{booktabs}
\usepackage{amsmath, amssymb, amsthm, thmtools, amsfonts}
\usepackage{bbm}
\usepackage[utf8]{inputenc}
\usepackage{hyperref}

\usepackage[margin=0.97in]{geometry}

\usepackage{ifthen}
\usepackage{tikz}
\usetikzlibrary{positioning,decorations.pathreplacing}
\usepackage{appendix}
\usepackage{graphicx}
\usepackage{adjustbox}
\usepackage{array}
\usepackage{makecell}
\usepackage{pifont}
\usepackage{amsmath}

\usepackage{color}
\usepackage{pagecolor}
\usepackage[boxed]{algorithm2e}
\usepackage[noend]{algpseudocode}
\usepackage{epstopdf}
\usepackage[textsize=tiny]{todonotes}

\newcommand\eps{\varepsilon}

\usepackage{framed}
\usepackage[framemethod=tikz]{mdframed}
\usepackage[bottom]{footmisc}
\usepackage[shortlabels]{enumitem}
\setitemize{noitemsep,topsep=3pt,parsep=3pt,partopsep=3pt}
\usepackage[font=small]{caption}
\usepackage{xspace}

\newtheorem{theorem}{Theorem}[section]
\newtheorem{lemma}[theorem]{Lemma}

\newtheorem{meta-theorem}[theorem]{Meta-Theorem}
\newtheorem{claim}[theorem]{Claim}
\newtheorem{remark}[theorem]{Remark}
\newtheorem{corollary}[theorem]{Corollary}

\newtheorem{definition}[theorem]{Definition}

\newtheorem*{t:node_compression}{Theorem~\ref{t:node_compression}}
\newtheorem*{t:edge_compression}{Theorem~\ref{t:edge_compression}}
\newtheorem*{l:mc_sequential}{Lemma~\ref{l:mc_sequential}}
\newtheorem*{l:mc_sequential2}{Lemma~\ref{l:mc_sequential2}}
\newtheorem*{l:mc_pram}{Lemma~\ref{l:mc_pram}}
\newtheorem*{l:mc_mpc}{Lemma~\ref{l:mc_mpc}}
\newtheorem*{l:mc_congest}{Lemma~\ref{l:mc_congest}}

\definecolor{darkgreen}{rgb}{0,0.5,0}
\definecolor{darkred}{rgb}{0.9,0,0}
\definecolor{grayish}{rgb}{0.8,0.8,0.8}
\usepackage{hyperref}
\hypersetup{
    unicode=false,          
    colorlinks=true,        
    linkcolor=darkred,          
    citecolor=darkgreen,        
    filecolor=magenta,      
    urlcolor=cyan           
}
\usepackage[capitalize, nameinlink]{cleveref}
\crefname{theorem}{Theorem}{Theorems}
\Crefname{lemma}{Lemma}{Lemmas}
\Crefname{claim}{Claim}{Claims}

\algnewcommand\algorithmicswitch{\textbf{switch}}
\algnewcommand\algorithmiccase{\textbf{case}}

\algdef{SE}[SWITCH]{Switch}{EndSwitch}[1]{\algorithmicswitch\ #1\ \algorithmicdo}{\algorithmicend\ \algorithmicswitch}%
\algdef{SE}[CASE]{Case}{EndCase}[1]{\algorithmiccase\ #1}{\algorithmicend\ \algorithmiccase}%
\algtext*{EndSwitch}%
\algtext*{EndCase}%

\newcommand{\bigO}{\mathcal{O}}

\renewcommand{\paragraph}[1]{\vspace{0.15cm}\noindent {\bf #1}:}

\mathchardef\mhyphen="2D
\newcommand{\up}[1]{^{^{\vbox{\hbox{$\scriptstyle#1$}\nointerlineskip\hbox{}}}}}

\newcommand{\MC}{$\mathsf{Minimum}\ \mathsf{Cut}$\xspace}

\newcommand{\MPC}{$\mathsf{MPC}$\xspace}
\newcommand{\PRAM}{$\mathsf{PRAM}$\xspace}
\newcommand{\CPRAM}{$\mathsf{CREW}$\xspace$\mathsf{PRAM}$\xspace}

\newcommand{\congest}{$\mathsf{CONGEST}$\xspace}

\newcommand{\poly}{\operatorname{\text{{\rm poly}}}}

\newcommand{\expval}[1]{E\left[#1\right]}
\newcommand{\Prob}[1]{P\left(#1\right)}

\newcommand{\whp}{with high probability}

\newcommand{\set}[1]{\left\{#1\right\}}
\newcommand{\paren}[1]{\mathopen{}\left(#1\right)\mathclose{}}

\newcommand{\floor}[1]{\left\lfloor #1 \right\rfloor\mathclose{}}
\newcommand{\card}[1]{\left|#1\right|}

\renewcommand{\paragraph}[1]{\vspace{0.15cm}\noindent {\bf #1}:}

\newcommand{\FullOrShort}{full}

\ifthenelse{\equal{\FullOrShort}{full}}{

	\newcommand{\fullOnly}[1]{#1}
  \newcommand{\shortOnly}[1]{}
  
  }{

    \newcommand{\fullOnly}[1]{}
    \newcommand{\IncludePictures}[1]{}
   
  }

\begin{document}

\date{}

\title{Faster Algorithms for Edge Connectivity \\ via Random $2$-Out Contractions}

\author{\emph{Mohsen Ghaffari}\thanks{Mohsen Ghaffari's research is supported by the Swiss National Foundation, under project
number 200021\_184735.}\\
\small ETH Zurich\\
\small \texttt{ghaffari@inf.ethz.ch}
\and
\emph{Krzysztof Nowicki}\thanks{Krzysztof Nowicki's research is supported by the National Science Centre, Poland grant 2017/25/B/ST6/02010.}\\
\small Univ. of Wroclaw\\
\small \texttt{knowicki@cs.uni.wroc.pl}
\and
\emph{Mikkel Thorup}\thanks{
Mikkel Thorup's research is supported by 
Investigator Grant 16582, Basic Algorithms Research Copenhagen (BARC), from the VILLUM Foundation.}\\
  \small BARC, Univ. of Copenhagen\\
  \small \texttt{mikkel2thorup@gmail.com}}

\setcounter{page}{0}
\thispagestyle{empty}

\maketitle
\begin{abstract}
We provide a simple new randomized contraction approach to the global minimum cut problem for simple undirected graphs. The contractions exploit 2-out edge sampling from each vertex rather than the standard uniform edge sampling. We demonstrate the power of our new approach by obtaining better algorithms for sequential, distributed, and parallel models of computation. Our end results include the following randomized algorithms for computing edge connectivity, with high probability\footnote{We use the phrase \emph{with high probability} (whp) to indicate that a statement holds with probability $1-O(n^{-\gamma})$, for any desired constant $\gamma \geq 1$.}: 

\smallskip
\begin{itemize}
\item Two \emph{sequential} algorithms with complexities $O(m\log n)$ and $O(m+n\log^3 n)$.  These improve on a long line of developments including a celebrated $O(m\log^3 n)$ algorithm of Karger [STOC'96] and the state of the art $O(m \log^2 n (\log\log n)^2)$ algorithm of Henzinger et al. [SODA'17]. Moreover, our $O(m+n\log^3 n)$ algorithm is optimal when $m=\Omega(n\log^3 n)$.
\smallskip
\item An $\tilde{O}(n^{0.8} D^{0.2} + n^{0.9})$ round \emph{distributed} algorithm, where $D$ denotes the graph diameter. This improves substantially on a recent breakthrough of Daga et al.[STOC'19], which achieved a round complexity of $\tilde{O}(n^{1-1/353}D^{1/353} + n^{1-1/706})$, hence providing the first sublinear distributed algorithm for exactly computing the edge connectivity. 
\smallskip
\item The first $O(1)$ round algorithm for the \emph{massively parallel computation} setting with linear memory per machine.  
\end{itemize}
\end{abstract}

\bigskip
\setcounter{page}{0}
\thispagestyle{empty}
\newpage
\tableofcontents
\setcounter{page}{0}
\thispagestyle{empty}
\newpage

%
%

\section{Introduction and Related Work}
\vspace{-5pt}
Computing the minimum cut is one of the classic graph problems, with a range of applications --- e.g., in analyzing the failure robustness of a network or in identifying the communication bottlenecks --- and it has been studied extensively since the 1960s~\cite{Ford-Fulkerson_algo, ford1962flows, gomory1961multi}. Of particular interest to the present paper is the case of simple undirected graphs: here, the objective is to identify the smallest set of edges whose removal disconnects the graph. This is often called the \emph{edge connectivity} problem. Our main contribution is to propose a simple randomized contraction process that, when combined with some other ideas, leads to faster algorithms for the edge connectivity problem in a number of computational settings. Next, in \Cref{subsec:related}, we overview some of the previous algorithmic developments, in \Cref{subsec:results} we state our improvements for various computational settings, and in \Cref{subsec:novelty}, we provide a brief overview of this randomized contraction process.

\vspace{-5pt}
\subsection{Related Work}
\label{subsec:related} Here, we discuss three lines of developments from prior work that are most directly related to our work. Some other results that are relevant as a point of comparison for our algorithms will be mentioned later, when describing our particular results for different computational settings. Moreover, we refer to \cite[Section 1]{DBLP:conf/stoc/KawarabayashiT15} for a nice survey of other work on this problem. 

(I) In 1993, Karger presented his random contraction idea for minimum cut~\cite{Karger1993ContracitonAlgorithm}: contract randomly chosen edges, one by one, until only two nodes remain. The edges in between are the minimum cut with probability $\Omega(1/n^2)$. Thus, via $O(n^2 \log n)$ repetitions, we can identify the minimum cut, with high probability. Thanks to its extreme simplicity and elegance and a range of corollaries, this  has by now become a well-known result and a standard ingredient of many textbooks and classes on algorithms. Shortly after, Karger and Stein~\cite{DBLP:conf/stoc/KargerS93} presented a bootstrapped version of this contraction idea, which sets up a recursion that stops the contractions processes at some point,  and judiciously chooses how many times re-run each of them, which gives an algorithm for the minimum cut problem with time complexity $O(n^2 \log^3 n)$. 

(II) In 1996, Karger~\cite{DBLP:conf/stoc/Karger96} provided the first algorithm for minimum cut that has a near-linear complexity in the number of edges; it runs in $O(m\log^3 n)$ time. This algorithm uses a different approach: it performs a certain packing of spanning trees, a la Tutte and Nash-Williams~\cite{Tutte, Nash-Williams}, and then reads the cuts defined by removal of any two edges from a tree, and reports the minimum such cut. This near-linear time algorithm is randomized (and Monte Carlo) and the question of obtaining a deterministic near-linear time algorithm (or even a near-linear time Las Vegas randomized algorithm) remained open for a long time. 

(III) In 2015, Kawarabayashi and Thorup~\cite{DBLP:conf/stoc/KawarabayashiT15} gave the first such deterministic algorithm, which in $O(m\log^{12} n)$ time computes a minimum cut in simple graphs (i.e., solves the edge connectivity problem). Their key new idea was to exploit that in simple graphs, all non-singleton min-cuts are very sparse, and hence we can contract all edges that are not in sparse cuts. These contractions lead to a graph with $\tilde{O}(m/\delta)$ edges, where $\delta$ denotes the minimum degree, while preserving all non-singleton\footnote{A cut in which exactly one node is on one side and all the other nodes are on the other side is called a \emph{singleton} or a \emph{trivial} cut. It is trivial to read the size of all singleton cuts --- i.e., node degrees --- and identify their minimum. } minimum cuts. This sparser (multi-)graph is then solved using older and slower algorithms. To find the sparse cuts,  Kawarabayashi and Thorup~\cite{DBLP:conf/stoc/KawarabayashiT15} used a deterministic near-linear time diffusion-based algorithm, inspired by page rank [PBMW99]. Later, Henzinger, Rao, and Wang~\cite{DBLP:conf/soda/HenzingerRW17} provided a faster deterministic algorithm following a similar framework, but based on computing flows to find sparse cuts. Overall, this algorithm computes the minimum cut in $O(m\log^2 n (\log\log n)^2 )$ time. This improved on the $O(m\log^3 n)$ algorithm of Karger~\cite{DBLP:conf/stoc/Karger96} and is the state of the art time complexity for edge connectivity. We remark that Karger's algorithm~\cite{DBLP:conf/stoc/Karger96} works also for weighted graphs, while those of \cite{DBLP:conf/stoc/KawarabayashiT15}, \cite{DBLP:conf/soda/HenzingerRW17}, and ours are limited to simple unweighted graphs (i.e., the edge connectivity problem).


\subsection{Our Results} 
\label{subsec:results}
In the next four subsubsections, we overview the algorithmic improvements that we obtain for computing edge connectivity in various computational settings.  \Cref{tab:previous_results1,tab:previous_results2} summarize the previous state of the art as well as our algorithms. As a rough summary, we note that our improvement is more substantial quantitatively for settings of distributed and massively parallel computation --- which includes a \emph{polynomially improved distributed algorithm} and the \emph{first constant time massively parallel algorithm} with linear local memory as detailed in \Cref{tab:previous_results2}. In contrast, the sequential improvements --- which includes an \emph{optimal sequential algorithm} for graphs with at least $n \log^3 n$ edges, as detailed in \Cref{tab:previous_results1} --- is probably accessible and interesting for a broad range of readers. 


\subsubsection{Improvements in Sequential Algorithms}
\label{subsubsec:seq}
\vspace{-5pt}
\paragraph{Our Contribution} For the sequential setting, our main end-result are combinatorial algorithms that provide the following statement:
\begin{theorem} \label{t:mc_sequential}
Given a simple input graph $G$, with $n$ vertices and $m$ edges it is possible to find its minimum cut in $\min\{\bigO(m + n \log^3 n), \bigO(m\log n)\}$ time, \whp. Moreover, we can obtain a cactus representation of all the minimum cuts in $\min\{\bigO(m + n \log^{O(1)} n), \bigO(m\log n)\}$ time. 
\end{theorem}
The corresponding proof appears in \Cref{thm:Gabow,thm:fasterCut}. For minimum cut, the $O(m+n\log^3 n)$ bound ensures that the algorithm has an optimal complexity whenever $m=\Omega(n\log^3 n)$. The $O(m\log n)$ part of the complexity bound is interesting for sparser graphs. Moreover, it improves on the state of the art $O(m \log^2 n \cdot (\log\log n)^2)$ algorithm \cite{DBLP:conf/soda/HenzingerRW17}. 


\subsubsection{Improvements in Distributed Algorithms}
\label{subsubsec:dist}
\vspace{-5pt}
\paragraph{Setting} 
We use the standard message-passing model of distributed computing (i.e., $\mathsf{CONGEST}$~\cite{Peleg:2000}): there is one processor on each graph node, which initially knows only its own edges, and per round each processor can send one $O(\log n)$-bit message to each of its neighbors. At the end, each processor should know its own part of the output, e.g., which of its edges are in the identified cut.   

\paragraph{State of the Art} Recently, Daga et al.~\cite{daga2019} gave the first distributed algorithm that computes the exact min-cut in simple graphs in sublinear number of rounds. Their algorithm runs in $\tilde{O} (n^{1-1/353}D^{1/353}+n^{1-1/706})$ rounds, where $D$ denotes the network diameter. In contrast, for graphs with a small edge-connectivity $\lambda$, a sublinear-time algorithm was known due to Nanongkai and Su~\cite{Nanongkai2014-distributed-cut}, with round complexity $\tilde{O}((D+\sqrt{n})\lambda^4)$. They also provide a $1+\eps$ approximation for any constant $\eps >0$, which runs in $\tilde{O}(D+\sqrt{n})$ rounds, and was an improvement on a $2+\eps$ approximation of Ghaffari and Kuhn~\cite{distributed-cut} with a similar round complexity. 

\paragraph{Our Contribution} We obtain an algorithm that runs in $\tilde \bigO \paren{n^{0.8}D^{0.2}+n^{0.9}}$ rounds, which provides a considerable improvement on the barely sublinear complexity of Daga et al.~\cite{daga2019}. The overall algorithm is also considerably simpler. 
\begin{theorem}\label{t:mc_congest}
  Given a simple input graph $G$, with $n$ vertices and $m$ edges it is possible to find its minimum cut in the \congest model in $\tilde \bigO \paren{n^{0.8}D^{0.2}+n^{0.9}}$ rounds, \whp. 
\end{theorem}


\begin{table}[t]
\begin{adjustbox}{center}
\begin{tabular}{| *{4}{c|} }
    \hline
\multicolumn{1}{|c|}{}
        & \multicolumn{1}{c|}{Sequential}
                & \multicolumn{2}{c|}{Classic Parallel \PRAM} \\
    \hline
    & work & work & depth \\
\Xhline{2\arrayrulewidth}

\makecell{\small simple graphs, \\ \small previous results}             
&   $\bigO(m \log^2 n (\log \log n)^2)$ \cite{DBLP:conf/soda/HenzingerRW17}
&  $\bigO(m \log^4 n)$     & $\bigO(\log^3 n)$\cite{Geissmann2018_PRAM_MC} \\
    \hline

\makecell{simple graphs, \\ \small\emph{*\textbf{our results}*}}             
&   $\min\{\bigO(m\log n),\bigO(m + n \log^3 n)\}$ 
&   $\bigO(m \log n + n \log^4 n)$     & $\bigO(\log^3 n)$ \\
    \Xhline{2\arrayrulewidth}

\end{tabular}
\end{adjustbox}
\vspace{-5pt}
\caption{\footnotesize Comparison of results, in the Sequential and PRAM settings. See \Cref{subsubsec:seq,subsubsec:PRAM}.}
\label{tab:previous_results1}
\end{table}
\vspace{-7pt}

\subsubsection{Improvements in Massively Parallel Algorithms}
\label{subsubsec:MPC}
\vspace{-5pt}
\paragraph{Setting} Parallel algorithms and especially those for modern settings of parallel computation (such as MapReduce~\cite{dg04}, Hadoop~\cite{White:2012}, Spark~\cite{ZahariaCFSS10}, and Dryad~\cite{Isard:2007}) have been receiving increasing attention recently, due to the need for processing large graphs. We work with the Massively Parallel Computation (MPC) model, which was introduced by Karloff et al.~\cite{KarloffSV10} and has by now become a standard theoretical model for the study of massively parallel graph algorithms.

In the MPC model, the graph is distributed among a number of machines. Each machine has a limited memory $S$ --- known as the \emph{local memory} --- and thus can send or receive at most $S$ words, per round. The number of machines $M$ is typically assumed to be just enough to fit all the edges, i.e., $O(m/S)$ or slightly higher. We refer to $M\cdot S$ as the \emph{global memory}. The main measure is the number of rounds needed to solve the problem, given a predetermined limited local memory. 

\paragraph{State of the Art} In the super-linear regime of local memory where $S=n^{1+\eps}$ for some constant $\eps>0$, many graph problems ---particularly, including minimum cut~\cite{lattanzi2011filtering} --- can be solved in $O(1)$ rounds, using a relatively simple filtering idea. Much of the recent activities in the area has been on achieving similarly fast algorithms for various problems in the much harder memory regimes where $S$ in nearly linear or even sublinear in $n$~\cite{czumaj2017round, ghaffari2018improved, boroujeni2018approximating,andoni2018parallel, assadi2019coresets, ghaffari2019sparsifying, behnezhad2019TreeMIS, brandt2019breaking, assadi2019massively, gamlath2018weighted, chang2019coloring, behnezhad2019exponentially, ghaffari2019conditionalLB}.

For minimum cut, in the nearly linear memory regime where $S=\tilde{O}(n)$ regime, the result given by Lattanzi et al. runs in $O(\log^2 n)$ rounds~\cite{lattanzi2011filtering} and requires global memory of order $\bigO(mn)$. It seems to be that the running time of this approach could be improved (by providing better implementation of the contraction process), but its global memory requirement is always $\Omega(n^2)$.

\paragraph{Our Contribution} We give the first algorithm with $O(1)$ round complexity while using only $\bigO(n)$ memory per machine and $\bigO(m + n \log^3 n)$ global memory. This settles the complexity of minimum cut in the nearly-linear memory regime.

\begin{theorem}\label{t:mc_mpc}
  Given a simple input graph $G$, with $n$ vertices and $m$ edges it is possible to find its minimum cut in $\bigO(1)$ rounds, \whp, using $\bigO(n)$ local memory per machine and $\bigO(m + n \log^3 n)$ global memory.
\end{theorem}

\subsubsection{Improvements in PRAM Parallel Algorithms}
\label{subsubsec:PRAM}
For the standard \PRAM model of parallel algorithms (Concurrent Write Exclusive Read), our algorithm improves the total work while achieving the same depth complexity as the state of the art~\cite{Geissmann2018_PRAM_MC}. We get an algorithm with $\bigO(\log^3 n)$ depth and $\bigO(m\log n + n\log^4 n)$ work. This improves on the work complexity of the state of the art algorithm of Geissman and Gianinazzi~\cite{Geissmann2018_PRAM_MC}, which has $\bigO(\log^3 n)$ depth and $\bigO(m\log^4 n)$ work.
\begin{theorem}\label{t:mc_pram}
There exists a \CPRAM algorithm that for a simple graph with $n$ vertices and $m$ edges computes its minimum cut using $\bigO(m \log n + n \log^4 n)$ work, with $\bigO(\log^3 n)$ depth. The algorithm returns a correct answer \whp.
\end{theorem}


\begin{table}[t]
\begin{adjustbox}{center}
\begin{tabular}{| *{5}{c|} }
    \hline
\multicolumn{1}{|c|}{}
                        & \multicolumn{3}{c|}{Modern Parallel \MPC}                
                                & \multicolumn{1}{c|}{Distributed \congest}                \\
    \hline
    & \makecell{local \\ memory} & \makecell{total \\ memory} & rounds & rounds   \\
    \Xhline{2\arrayrulewidth}
    
\makecell{\small simple  graphs, \\ \small previous results}             
&   \makecell{$\bigO(n)$ \\ $\bigO(n^{1+\eps})$}  &  \makecell{$\bigO(mn)$ \\  $\bigO(nm)$}  &  \makecell{$\bigO(\log^2 n)$ \cite{lattanzi2011filtering} \\ $\bigO(1)$ \cite{lattanzi2011filtering} } 
& \makecell{$\tilde \bigO (n^{1-1/353}D^{1/353}+n^{1-1/706})$ \\ ~\cite{daga2019} }    \\
    \hline

\makecell{\small simple  graphs, \\ \small \emph{*\textbf{our results}*}}             
&   $\bigO(n)$  &  $\bigO(m + n \log^3 n)$  &  $\bigO(1)$  
& $\tilde \bigO (n^{0.8}D^{0.2}+n^{0.9})$       \\
    \Xhline{2\arrayrulewidth}

\end{tabular}
\end{adjustbox}
\vspace{-5pt}
\caption{\footnotesize Comparison of results, in the massively parallel computation and distributed computation settings. The previous results for the MPC model show two lines, for different regimes of local memory. See \Cref{subsubsec:dist,subsubsec:MPC}.}
\label{tab:previous_results2}
\vspace{-5pt}
\end{table}


\vspace{-7pt}
\subsection{Our Method, In a Nutshell}
\label{subsec:novelty}
\vspace{-5pt}

Our main technical contribution is a simple, and plausibly practical, randomized contraction process that transforms any $n$-node graph with minimum cut $\lambda$ to a multi-graph with $O(n/\lambda)$ vertices and only $O(n)$ edges, while preserving all non-singleton\footnote{A cut in which exactly one node is on one side is called a \emph{singleton} cut, or sometimes a \emph{trivial} cut. In most computational settings, it is trivial to read the size of all singleton cuts --- that is, degrees of nodes --- and identify their minimum. Thus, the problem of computing the edge connectivity effectively boils down to assuming the minimum cut size is at most the minimum degree and identifying the smallest non-singleton cut.} minimum cuts with high probability. This can also be viewed as a simple \emph{graph compression} for (non-trivial) minimum cuts. We then solve the minimum cut problem on this remaining sparse (multi-)graph, using known algorithms. 

The aforementioned contraction process itself has two parts, and some careful repetition for success amplification, as we overview next. (A) The main novelty of this paper is the first contraction part, which we refer to as \emph{random $2$-out contraction}: for each node $v$, we randomly choose $2$ of its edges (with replacement and independent of other choices) --- we view these as ``outgoing" edges from node $v$, proposed by $v$ for contraction --- and we contract all chosen edges simultaneously. We prove that this reduces the number of vertices to $O(n/\lambda)$ while preserving any non-trivial minimum cut with a constant probability. In fact, we show the number of vertices is in $O(n/\delta) \leq O(n/\lambda)$, where $\delta$ denotes the minimum degree. Furthermore, the contraction preserves any non-singleton cut with size at most $2-\eps$ factor of the minimum cut size, with a constant probability, for any constant $\eps \in (0, 1]$. (B) For the second part, we transform the graph after the first part of contractions to have also only $O(n)$ edges---i.e., within an $O(\lambda)$ factor of the number of vertices---while preserving any cut of size $O(\lambda)$. There are several ways to obtain that goal, in this paper we discuss a deterministic approach based on sparse connectivity certificates and a randomized approach based on contracting a uniformly sampled subset of edges. (C) Finally, we use $O(\log n)$ repetitions of the combination of these two parts, and a carefully designed ``majority" voting per edge, to  amplify the success probability and conclude that with high probability, all non-trivial minimum cuts are preserved, while having $O(n/\delta) \leq O(n/\lambda)$ vertices and $O(n)$ edges.

\paragraph{Related Work on $k$-Out} We note that random $k$-out subgraphs have been studied in the literature of random graphs. See for instance the work of Frieze et al.~\cite{frieze2017random} which shows that if $\delta \geq (1/2+\epsilon) n$, then a $k$-out subgraph (to be precise, sampling $k$ edges per node and without replacement) is whp $k$-connected, for $k=O(1)$. We leverage a somewhat opposite property of $k$-out: that with some constant probability, it does not contract any edge for a singleton minimum cut (hence, the subgraph is not even connected), while still significantly reducing the number of \emph{vertices}. 

In a very recent paper \cite{HKTZZ19:k-out}, Holm et al. have shown
that if $k\geq c\log n$, then the $k$-out contraction of a simple
graph with $n$ vertices has only $O(n/k)$ edges (their $k$-out
definition is slightly different, but easily converted to ours). If
such a result was true for $k=2$, then this would simplify our
constructions, but proving it for close to constant $k$ seems far out
of reach with current techniques.  The probability of destroying a
min-cut with a $k$-out sample grows exponentially in $k$, so the
techniques from \cite{HKTZZ19:k-out} are not relevant to our min-cut
computation. Conversely, the results presented here have no
impact on the targets from \cite{HKTZZ19:k-out}. After $2$-out sampling,
we do contract edges with highly connected end-points to get down
to $O(n)$ edges, but this is only valid because we only care about
small cuts. In short, the only relation to \cite{HKTZZ19:k-out} is
that both our work and theirs study and show algorithmic benefits of $k$-out sampling.

\paragraph{Roadmap} In \Cref{sec:contraction}, we describe our contraction process and state its guarantees. Later, in \Cref{s:review}, we outline how by using this contraction process and some other algorithmic ideas, we obtain faster min cut algorithms for various computational settings. The details of the implementations in various computational settings appear later, in separate sections. 

\section{Our Contraction Process}
\label{sec:contraction}
\vspace{-7pt}

\paragraph{Basic Definitions and Notations} We are working with a simple graph and we use $n$ to denote its number of vertices, $m$ to denote its number of edges, $\delta$ to denote its minimum degree, and $\lambda$ to denote its edge connectivity, i.e., the smallest number of edges whose removal disconnect the graph. When dealing with different graphs, we may use subscript
notation to say which graph we are working with. For a graph $H$,
we let $n_H$ denote the number of nodes, $m_H$ the number of edges,
$\delta_H$ the minimum degree, and $\lambda_H$ the edge connectivity of
$H$.

We define a cut in a graph by the set of vertices that are on one side of this cut. If a cut is defined by some set of vertices $S\subset V$, the edges of the cut are the edges of the graph that have exactly one endpoint in $S$. Furthermore, we denote the set of edges of cut $S$ by $C(S)$, and the size of a cut $C(S)$ by $\card{C(S)}$. We call the cut $C(S)$ a non-singleton iff $\card{S} > 1$. We say that $C(S)$ is a minimum cut, if for each cut $S'$, we have $\card{C(S)} \leq \card{C(S')}$. For a given value $\alpha\geq 1$, we say that a cut $C(S)$ is an $\alpha$ minimum cut, or $\alpha$-small, iff $\card{C(S')} \leq \alpha\lambda$. 

For any edge set $D$ from $G$, we denote by $G/D$ the result of
contracting the edges from $D$ in $G$. In this paper, we identify a
cut with the cut edges connecting the two sides. The contraction of
$D$ preserve a given cut $C$ if and only if $C\cap D=\emptyset$. The
understanding here is that all edges have identifiers, that is, they
are not just vertex pairs, so when we do contractions and remove self-loops, it is well-defined which edges survived.

\subsection{Contraction Outline}
\label{subsec:contraction-outline}
\medskip

  Our main result is captured by the following statement. While stating this result, to make things concrete, we also mention the sequential time related for implementing it. The complexity for other computational settings is discussed in the later sections. 
	\begin{theorem}\label{thm:main} Let $G$ be a simple graph with $m$ edges, $n$ nodes, and
min-degree $\delta$. Fix an arbitrary constant $\eps \in (0,1]$. In $O(m \log n)$ time, we can randomly contract
the graph to a multi-graph $\widehat{G}$ with $O(n/\delta)$ nodes and $O(n)$ edges such that, whp, $\widehat{G}$ preserves all
non-trivial $(2-\eps)$-min-cuts of $G$.
\end{theorem}

At the heart of the above result is a contraction captured by \Cref{t:main_theorem} which preserves each particular small cut with a \emph{constant} probability. We discuss later in \Cref{subsec:succAmp}	how we \emph{amplify the success probability} so that we preserve all nearly minimum cuts whp, thus giving the above theorem, without sacrificing the number of nodes or edges.

  \begin{theorem}\label{t:main_theorem} Let $G$ be a simple graph with $m$ edges, $n$ nodes, and
min-degree $\delta$, and fix an arbitrary constant $\eps \in (0,1]$. Then, in $O(m)$ time, we can randomly contract
the graph down to $O(n/\delta)$ nodes and $O(n)$ edges such that, for any fixed non-trivial $(2-\eps)$-min-cut, we preserve the cut with at least a constant probability $p_\eps > 0$.
  \end{theorem}

\paragraph{Outline of the Contraction Process of \Cref{t:main_theorem}}
    Our contraction process has two parts. The first part is contracting a random $2$-out and, as formally stated in \Cref{t:node_compression}, we show that this step reduces the number of vertices to $O(n/\delta)$. The second part, stated in \cref{lem:NI}, reduces the number of edges to $O(n)$. Furthermore, in \cref{s:mc_mpc} we provide an alternative approach to the second part, that has an efficient parallel implementation. Each of these processes preserves any particular non-singleton $(2-\eps)$ minimum cut with at least a constant probability. Hence, their composition has a constant probability of preserving that cut. After both contraction processes, we have $\bigO\paren{\frac{n}{\delta}} = \bigO(\frac{n}{\lambda})$ nodes and $\bigO\paren{\frac{n}{\delta} \cdot \lambda } = \bigO(n)$ edges. 

		We comment that for our distributed algorithm, we actually do not need the part about reducing the number of edges. However, instead, we desire and prove another nice property from the random out contractions: that the summation of the diameters of the $2$-out is $\bigO\paren{\frac{n \log \delta}{\delta}}$. In fact, we show in \Cref{sec:mc_congest} that by choosing a subset of the edges of the $2$-out, we can define $\tilde{\bigO}\paren{\frac{n}{\delta}}$ components, each with $O(\log^2 n)$ diameter (clearly, contracting this subset of $2$-out preserves any cut, if the full set of that $2$-out preserved it).   
		
		We next discuss the two parts of the contraction process for reducing the number of vertices and edges, separately, in the next two subsections. In the last subsection of this section, we discuss how we amplify to success to preserve all non singleton $(2-\eps)$ minimum cuts.

\vspace{-5pt}  
\subsection{Reducing the number of vertices} \label{s:node_compression}
Here, we propose and analyze an extremely simple contraction process: each node proposes $k$ randomly sampled incident edges and we contract all proposed edges. More formally, we contract all connected components of a subgraph obtained by a random selection of $k$ incident edges for each vertex of the graph (sampled independently from the original graph, with repetitions). We call this random subgraph a random \emph{$k$-out subgraph}, and we call the related contraction process a \emph{random $k$-out contraction}. We show that a random $2$-out contraction reduces the number of nodes to $O(n/\delta)$ whp, while preserving any fixed nontrivial nearly minimum cut with a constant probability.
%
%
%
 %
  Next, we formalize the notion of random $k$-out contractions, and prove their properties.
  
  \subsubsection{Random $k$-out contractions}
  \begin{definition}
    Let consider a graph $G=(V,E)$. Let $I_j$, for $j \in \set{1, \dots, k}$, be a set of edges obtained by selecting for each node a uniformly random incident edge. 
    \begin{itemize}
      \item we call a graph $G^{(k)} = (V,\bigcup_{j=1}^k I_j)$ a random $k$-out subgraph of $G$
      \item a random $k$-out contraction of a graph $G$, is a multigraph obtained by contracting connected components of $G^{(k)}$
      \item we use the phrase ``cut $C(S)$ is preserved by the random $k$-out contraction" to indicate that $C(S) \cap \bigcup_{j=1}^k I_j = \emptyset$. 
    \end{itemize}
  \end{definition}
  
  \noindent Firstly, we show the following properties of random $2$-out contractions. 
  \begin{theorem}\label{t:node_compression}
    A random $2$-out contraction of a graph with $n$ vertices and minimum degree $\delta$ has $\bigO(\frac{n}{\delta})$ vertices, \whp, and preserves any fixed non singleton $(2-\eps)$ minimum cut, for any constant $\eps \in (0, 1]$, with some constant probability at least $p_\eps > 0$.
  \end{theorem}
  \begin{proof}[Proof Outline]
    The proof consists of two parts, which are presented in two separate subsubsections. First, in \Cref{s:number_of_components}, we prove that the number of nodes after contractions is $\bigO\paren{\frac{n}{\delta}}$, whp. Then, in \Cref{s:preserving_cut}, we show that a random $1$-out contraction preserves a fixed non-singleton $(2-\eps)$ minimum cut with probability that is at least some positive constant $q_{\eps}>0$. This implies that a $2$-out contraction (which is simply $2$ independent $1$-out contractions) has probability of preserving this cut at least $p_\eps = q_\eps^2>0$.
  \end{proof}
  
  \subsubsection{Number of nodes after 2-out contraction} \label{s:number_of_components}
	
      In this part of proving \Cref{t:node_compression}, we bound the number of nodes after a $2$-out contraction. 
      \begin{lemma}\label{l:node_compression_number_of_components}
        The number of connected components in a random $2$-out subgraph of a simple graph with $n$ vertices and the minimum degree $\delta$ is $\bigO\paren{\frac{n}{\delta}}$, \whp.
      \end{lemma}

      \begin{proof}
      We consider a random process that starts with a graph $G'=(V, \emptyset)$ and gradually adds to $G'$ the edges of a random $2$-out subgraph $G^{(2)}$. During this process, each vertex can be in one of three states: \emph{processed}, \emph{active}, or \emph{unprocessed}. The process is performed in phases. Each phase starts in an arbitrary unprocessed vertex $v$  (marked as \emph{active}) and builds a set of vertices reachable from $v$ (which also became \emph{active} as long as they are not \emph{processed}) by the random edges of active vertices. During the process, we maintain a counter $\kappa$ that is incremented only at the end of a phase that creates a new connected component in $G^{(2)}$, and only if that component is smaller than some threshold value $x$. The final bound on the number of connected components is the final value of $\kappa + \frac{n}{x}$.
     Basically, \emph{processed} vertices are the vertices that are in connected components that are taken into account either in the counter $\kappa$ or they are among at most $\frac{n}{x}$ components of size at least $x$. Thus, if at the end of a phase we can reach a \emph{processed} vertex from $v$, we did not create a connected component that is not already included in the bound, and we do not increase the counter $\kappa$.
      
      Let us consider following random variables $X_i$: if phase $i$ ended with increasing $\kappa$, $X_i = 1$, otherwise $X_i = 0$. In other words, $X_i$ takes value $1$ only if the $i$th phase ends with creating new connected component that is smaller than $x$.  
      Furthermore, if at the end of the phase, we can reach a processed vertex from $v$ then $X_i = 0$, therefore $$\Prob{X_i = 1 | X_1 X_2 \dots X_{i-1}} \leq \Prob{X_i = 1 | X_1 X_2 \dots X_{i-1} \wedge \text{all processed vertices not reachable from } v}.$$ 
      
			\smallskip
      \paragraph{Bound on the probability of incrementing counter $\kappa$}     
      Let $A_v$ denotes the set of active vertices in the current phase that started in $v$. We start with $A_v=\{v\}$, and then we repeat the following sampling events, one by one. If there is a vertex $u$ from $A_v$ where we have not generated all its sample edges, we pick the first such vertex $u$ added in $A_v$ and generate its next sample edge $(u,w)$ uniformly at random among the at least $\delta$ edges incident to $u$. We say that the sample is ``\emph{caught}'' if $w\in A_v$. We terminate when there are no more samples to do from $A_v$ and mark all vertices from $A_v$ as \emph{processed}.
      
      If $A_v$ ends up at final size $x$, we know we have performed exactly $2x$ sampling events, and that the samples were caught exactly $x+1$ times. The probability that a given sample is caught is at most $(x-1)/\delta$. The order in which we generate the samples is completely defined in the above process. There are ${2x\choose x+1}$ choices for placing the $x+1$ caught samples among all samples. The probability that we get this particular sequence of caught and not-caught samples is bounded by the probability that the subsequence of $x+1$ samples that are supposed to get caught actually get caught. This happens with probability at most $((x-1)/\delta)^{x+1}$. We conclude that the probability that we terminate with $|A_v|=x$ is upperbounded by $P_x={2x\choose x+1} ((x-1)/\delta)^{x+1}$. Then $P_1=0$ and $P_2=4/\delta^3$. Moreover, for $x\geq 2$, we have 
  
      \begin{align*}
        P_{x+1}/P_x&=\frac{{2x+2\choose x+2} (x/\delta)^{x+2}}{{2x\choose x+1} ((x-1)/\delta)^{x+1}}
        =\frac{(2x+2)(2x+1)}{(x+2)x} \left(\frac{x}{x-1}\right)^{x+1}\frac{x}{\delta}
        < 4 e^{\frac{x+1}{x-1}}\frac{x}{\delta}
        \leq \frac{4 e^3x}{\delta}.
      \end{align*}
  
      The last expression is bounded by $1/2$ for $x\leq x^*=\delta/(8e^3)$. Thus the probability that we end up with $|A_v|\leq x^*$ is bounded by
  
      \[P_{\leq x^*}=\sum_{x=2}^{x^*}P_x<2P_2\leq 8/\delta^3.\]
      
			\smallskip
      \paragraph{Value of $\kappa$ at the end of the process}    
    We have $$\Prob{X_i = 1 | X_1 \dots X_{i-1}}  \leq \Prob{X_i = 1 | X_1 \dots X_{i-1} \wedge \text{all processed vertices not reachable from } v} \leq  8/\delta^3 .$$ Hence, variables $ X_i | X_1 X_2 \dots X_{i-1} $ are stochastically dominated by independent random variables $Y_i$ that take value $1$ with probability $8/\delta^3$. We can conclude that $\Prob{\sum\limits_i X_i \leq \mu} < \Prob{\sum\limits_i Y_i > \mu}$ ~\cite[Lemma 8.7]{DBLP:journals/corr/abs-1801-06733}. By a Chernoff bound, for any $\mu \geq 8n/\delta^3 \geq \expval{\sum\limits_i Y_i}$, and for $\eps < 1$ we have
      
      \begin{equation}\label{eq:number}
      \Prob{\sum\limits_i X_i > (1+\eps)\mu} < \Prob{\sum\limits_i Y_i > (1+\eps)\mu} \leq \exp{\paren{-\frac{\mu\eps^2}{3}}}.
      \end{equation}
      while for $\eps \geq 1$ we have 
      \begin{equation}\label{eq:number2}
      \Prob{\sum\limits_i X_i > (1+\eps)\mu} < \Prob{\sum\limits_i Y_i > (1+\eps)\mu} \leq \paren{\frac{e^{\eps}}{(1+\eps)^{1+\eps}}}^{\mu}
      \end{equation}
      
      For $\delta \leq \sqrt{n}$, by \cref{eq:number}, the value of $\kappa$ is larger than $\frac{n}{\delta}$, with probability at most $\exp{\paren{-\Theta(\sqrt{n})}}$. For $\delta \geq \sqrt{n}$, $n / \delta^3 \leq 1/\sqrt n$, hence we can apply \cref{eq:number2} with $\mu = 1/\sqrt n $ and $1+\eps = \sqrt{n}$, which gives that the value of $\kappa$ is larger than some constant $c$ (which is always $\bigO(\frac{n}{\delta})$) with probability at most $(e/n)^{\frac{c}{2}}$.
      Therefore, for any $x \leq x^*$ the number of connected components of $G^{(2)}$ smaller than $x$ is $\bigO(\frac{n}{\delta})$, \whp. Thus, for $x=x^*$, we get the number of all connected components of $G^{(2)}$ is $\bigO(\frac{n}{\delta}) + \bigO(\frac{n}{x^*}) = \bigO(\frac{n}{\delta})$.
			\end{proof}
      
  \subsubsection{Preserving a fixed non singleton minimum cut}\label{s:preserving_cut}
    In this part of the proof of \Cref{t:node_compression}, we analyze the probability of preserving a fixed non singleton minimum cut $C(S)$. Before that, we recall a small helper inequality: 
		
		\begin{claim} \label[claim]{l:inequality}
        For any $x$ and $y$ such that $0 < x \leq y < 1$, we have $1-x > \paren{e^{-\frac{1}{1-y}}}\up{x}$
      \end{claim}
      \begin{proof}
        This inequality follows from the fact that for any $\alpha > 1$, $(1-\frac{1}{\alpha})^{\alpha-1} > e^{-1}$. If we set $x = \frac{1}{\alpha}$, we have $(1 - x)^{\frac{1}{x}-1} > e^{-1}$, which implies $1-x > e^{-\frac{1}{1-x}x} \geq \paren{e^{-\frac{1}{1-y}}}\up{x}$.
      \end{proof}
			
    \begin{lemma}
      \label{l:1_out_lower_bound}
      Probability of preserving a fixed non singleton $(2-\eps)$ minimum cut $C(S)$, for any constant $\eps \in (0, 1]$, by a random $1$-out contraction is at least some constant $q_\eps > 0$.
    \end{lemma}
    \begin{proof} 
      Let us denote by:
      \begin{itemize}
        \item $c(v)$ the number of edges incident to $v$ that are in $C(S)$,
        \item $d(v)$ the degree of a node $v$,
        \item $N(S)$ the set of nodes incident to the edges in $C(S)$. 
      \end{itemize}
      

\noindent The probability that a random $1$-out contraction does not contract any edge from $C(S)$ is 
$$ \prod\limits_{v \in N(S)} \paren{1 - \frac{c(v)}{d(v)}}. $$ 
To analyze this expression, we first argue that for any node $v \in N(S)$, $\frac{c(v)}{d(v)} \leq x < 1$, for some constant $x$. For that, let us denote by $\alpha = \card{C(S)}$. Then, the size of a cut defined by the set of vertices $S \setminus \set{v}$ is $\alpha + d(v) - 2c(v)$. On the one hand, we have that $C(S)$ is $(2-\eps)$ minimum cut, hence $\card{C(S \setminus \set{v})} \leq (2-\eps)\lambda + d(v) - 2c(v)$. On the other hand, $\card{C(S \setminus \set{v})} \geq \lambda$. Putting those two things together gives that $(1-\eps) \lambda + d(v) > 2c(v)$, which implies that $\frac{c(v)}{d(v)} \leq \frac{(1-\eps) \lambda + 1}{2d(v)} \leq 1-\frac{\eps}{2} $.

Now, we are ready to analyze $ \prod\limits_{v \in N(S)} \paren{1 - \frac{c(v)}{d(v)}}$, which is the probability that a random $1$-out contraction does not contract any edge from $C(S)$. We know that for each $v$ value of $\frac{c(v)}{d(v)}$ is upper bounded by $1-\eps/2 < 1$, hence we can use $\frac{c(v)}{d(v)}$ as $x$ and $1-\frac{\eps}{2}$ as $y$ in \Cref{l:inequality}. Let $z = e^{-\frac{1}{1-y}} = e^{-\frac{1}{\frac{\eps}{2}}}$. Notice that for any $\eps \in (0, 1]$, we have $z\in (0, e^{-2}] \subseteq (0, 0.2)$. Then, we have: 
      
      $$ \prod\limits_{v \in N(S)} \paren{1 - \frac{c(v)}{d(v)}} > z^{\sum\limits_{v \in N(S)}\frac{c(v)}{d(v)}}$$
      Furthermore, since a degree of a vertex can not be smaller than $\lambda$, we can conclude that the probability that we do not contract any edge from $C(S)$ is at least 
      $$z^{\sum\limits_{v \in N(S)}\frac{c(v)}{d(v)}} \geq z^{\sum\limits_{v \in N(S)}\frac{c(v)}{\lambda}} = z\up{2(2-\eps) \lambda / \lambda} = z\up{2(2-\eps)} = e^{-4(2-\eps)/\eps} > 0.
			$$
			\vspace{-9pt}
    \end{proof}
  

\subsection{Reducing the number of edges} \label{s:edge_compression}
  In this section, we explain how we reduce the number of edges in the graph resulting after $2$-out contractions --- which we now know to have $O(n/\lambda)$ vertices --- down to $O(n)$ edges, while preserving the minimum cut. Firstly, we present an a method based on sparse certificates and it preserves the minimum cut deterministically. Then, we propose an approach based on some variant of random contractions, which preserves the minimum cut with some constant probability and is easily implementable in parallel models of computing.
	
	\paragraph{Reducing the number of edges via sparse certificates}
We can reduce the number of edges in the contracted graph from \Cref{t:node_compression} using the
general $k$-edge connectivity certificate of Ibaraki and Nagamochi~\cite{IbarakiNagamochi1992}. They have shown an algorithm that, given a multigraph $M$ and a
number $k$, in linear time can find a subgraph $H$ with less than $kn$
edges so that if a cut has $c$ edges in $M$, then the cut between
same vertex sets has at least
$\min\{k,c\}$ edges in $H$. Cuts with at most $k$ edges in $M$ are thus
fully preserved in $H$. This way \emph{$H$ is a certificate of $k$-edge connectivity
in $M$}. Based on this, \cite{KT19:edge-conn} suggested contracting all
edges $e$ not in $H$. Since the end-points of an edge not
in $H$ must be $k$-edge connected in $M$, and thus $k+1$ connected in
$H$, contracting $e$ preserves any cut of size at most $k$. Summing
up, we get
\begin{lemma}[{\cite{KT19:edge-conn,IbarakiNagamochi1992}}]\label{lem:NI}  
Given a multigraph $M$
with $m_M$ edges and $n_M$ nodes and a number $k$, in $O(m_M)$ time, we
can identify and contract an edge set $D$ such that $M/D$ preserves
all cuts of size at most $k$ in $M$, and such that $M/D$ has at most
$n_M k$ edges.
\end{lemma}

Given a simple graph $G$, we can first
apply \Cref{t:node_compression} and then \Cref{lem:NI} with $k=2\delta\geq
2\lambda$, to get a proof of \cref{t:main_theorem}

\paragraph{Reducing the number of edges via random contractions}

\begin{lemma}\label{lem:random_edge_reduction}
	Given an unweighted multigraph $G = (V, E)$ with $n$ vertices, $m$ edges and minimum degree $\delta$, a contraction process that contracts a set of vertices $E_{1/(2\delta)}$, to which we include each edge of $E$ with probability $\frac{1}{2\delta}$, reduces the number of edges to $\bigO(n\delta)$ and preserves a fixed minimum cut with probability at least $\frac{1}{2}$.
\end{lemma}
\begin{proof}
\paragraph{The number of edges}
	Let $G_{1/(2\delta)} = (V, E_{1/(2\delta)})$. The number of edges of $G$ that are inter component in $G_{1/(2\delta)}$ is $\bigO(n\delta)$, with high probability. This statement follows directly from the analysis of the sampling approach to the MST problem \cite{Karger:1995:RLA:201019.201022}. The authors of \cite{Karger:1995:RLA:201019.201022} show slightly stronger claim, as they say that if $G$ is an $n$ node weighted graph, and $F$ is a minimum spanning forest of $G_{p}$, then the number of edges of $G$ that are $F$-light (edge $e$ is $F$ light if it is not the heaviest edge on a cycle in $F$ extended by $e$) is $\bigO(np)$ with high probability. Clearly, all edges of $G$ that are inter component in $G_p$ would be $F$-light, hence the number of inter component edges is smaller than the number of $F$-light edges. Furthermore, the analysis provided in \cite{Karger:1995:RLA:201019.201022} does not really use that the graph does not have parallel edges, which makes it applicable to our case.

\paragraph{Preserving a cut}
Le consider a fixed minimum cut of $G$, and let $\lambda$ be a size of this cut. Clearly $\lambda \leq \delta$. Therefore, probability that we don't include any edge of this cut in $E_{1/(2\delta)}$ is $\paren{1-\frac{1}{2\delta}}^{\lambda} \geq \paren{1-\frac{1}{2\lambda}}^\lambda \geq \frac{1}{2}$.

\end{proof}
		
\subsection{Amplifying success and preserving all nontrivial small cuts}
\label{subsec:succAmp}
We next prove \Cref{thm:main}, by using \Cref{t:main_theorem} and a careful ``\emph{repetition and voting}".

\begin{proof}[Proof of \Cref{thm:main}]
To prove this statement, we build a process that amplifies the success probability of \Cref{t:main_theorem} and preserves a
particular given non-trivial $(2-\eps)$-small cut $C$ of $G$ with high probability
$1-n^{-\gamma}$. Karger~\cite{Karger2000} has proved that the number of
$(2-\eps)$-small cuts is $O(n^3)$. Hence, by a union bound, we
conclude that all non-trivial $(2-\eps)$-small cut of $G$ are
preserved with high probability $1-1/n^{\gamma'}$ where
$\gamma'=\gamma-3$ is an arbitrarily large constant.

To build such a contraction with amplified success, we apply \Cref{t:main_theorem} for $q=O((\log n)/p_\eps)=O(\log n)$ times,
with independent random variables, yielding contracted multigraphs
$\widehat G_1,\ldots,\widehat G_q$. Each $\widehat G_i$ preserves our
non-trivial $(2-\eps)$-small cut $C$ with probability at least
$p_\eps$, so the expected number of $\widehat G_i$ that preserve $C$ is
at least $\mu=p_\eps q$.  Using a standard Chernoff bound (see, e.g.,
\cite[Theorem 4.2]{MR95}), the probability that less than $r=\mu/2=p_\eps
q/2$ of the $\widehat G_i$ preserve $C$ is upper bounded by 
$\exp(-\mu/8)=\exp(-p_\eps q)$. For any given $\gamma$, this 
is $O(n^{-\gamma})$ for $q\geq 8(\ln n){\gamma}/p_\eps$.

We now take each edge $e$ in $G$, and ask how many $\widehat G_i$ it is 
preserved in. Since each $\widehat G_i$ has $O(n)$ edges,
it can only preserve $O(n)$ edges. Therefore
the total number of edge preservation from all the $\widehat G_i$ 
is $q\cdot O(n)$. Therefore, the number of edges that are preserved at least $r$ times
is at most $q\cdot O(n)/r =O(\gamma n/p_\eps)=O(n)$. If
an edge is not preserved $r$ times, then, by assumption, it is
not in any non-trivial $(2-\eps)$-small cut, so we contract it. 

Since all but $O(n)$ edges of $G$ got contracted, the resulting 
graph $\widehat G$ has at most $O(n)$
edges. Moreover, our contractions did not contract any edge from any 
non-trivial $(2-\eps)$-min-cuts, whp. 
Finally, note that the original graph $G$ had min-degree $\delta$. As proved in
\cite{KT19:edge-conn}, if a cut of $G$ has size less than $\delta$, it
must have at least $\delta$ vertices on either side. Therefore, if a
node in $\widehat G$ has degree below $\delta$, then it must be
contracted from at least $\delta$ vertices in $G$, so we have at most
$n/\delta$ nodes in $\widehat G$ with degree below $\delta$. On the
other hand, we can have at most $2m_{\widehat G}/\delta=O(n/\delta)$ nodes in
$\widehat G$ with degree at least $\delta$. Hence, we conclude that the
total number of nodes in $\widehat G$ is $O(n/\delta)$.

Overall, we spend $O(mq)=O(m\log n)$ time, both on generating the
$q$ contracted multigraphs $\widehat G_i$ and on counting for
each edge in $G$ how many $\widehat G_i$ it is preserved in (since
surviving edges preserve their id doing contractions, we can with
each edge id record which $\widehat G_i$ it is preserved in).
\end{proof}
 
\section{General Algorithm Outline and Overview of Applications} \label{s:review}
  We now overview the applications of \Cref{thm:main} to various models of computing. On the very top level, all the algorithms we present fit the following description:
  \begin{enumerate}
    \item Compute the contraction $\widehat{G}$ of the input graph $G$ as indicated by \Cref{thm:main}.
    \item Compute the minimum cut of $\widehat{G}$ using an algorithm that works for multi-graphs.
    \item If the computed cut of $\widehat{G}$ is smaller than minimum degree of $G$, output it as a minimum cut of $G$. Otherwise, output the minimum degree of $G$ (and the corresponding vertex) as a (singleton) minimum cut.
  \end{enumerate}
  
\noindent Next, we give a brief description of our algorithms for different models. More detailed versions of these algorithms, as well as the formal definitions of the \PRAM, \MPC and \congest models, follow in the subsequent sections.
  \subsection{Sequential model}
	
	In \Cref{sec:Seq}, we give two sequential algorithms for computing edge connectivity with high probability, both of which follow the above outline. The first algorithm has a complexity of $O(m\log n)$ and follows rather directly from combining \Cref{thm:main} with running the minimum cut algorithm of Gabow~\cite{Gabow1991} on the contracted graph $\widehat{G}$. The detailed description is presented \Cref{subsec:firstSeq}.
	
	The second algorithm has a complexity of $O(m+ n \log^3 n)$. For this algorithm, in \Cref{subsec:secondSeq}, we present a more elaborate way of implementing a contraction similar to the one provided by \Cref{thm:main} but in just $O(m + n\log n \; \alpha(n, n))$ time, where $\alpha$ denotes the inverse Ackermann function. This process is based on a careful usage of the union-find data structure, and some other probabilistic ideas. Then, we solve the minimum cut problem on the resulting multi-graph, which has $O(n)$ edges, using Karger's algorithm~\cite{Karger2000} in $O(n\log^3 n)$ time, for a total complexity of $O(m + n\log n \; \alpha(n, n) + n\log^3 n) = O(m+n\log^3 n)$. We note that any improvement on Karger's algorithm for multi-graphs would immediately lead to an improvement in our algorithm.

  \subsection{\congest model}
    In \Cref{sec:mc_congest}, we provide a distributed algorithm that solves the edge connectivity problem in $\tilde{O}(n^{0.8} D^{0.2} + n^{0.9})$ rounds (as stated in \Cref{t:mc_congest}). This improves substantially on a recent breakthrough of Daga et al.~\cite{daga2019} that achieved the first sublinear round complexity, running in $\tilde{O} (n^{1-1/353}D^{1/353}+n^{1-1/706})$ rounds. Furthermore, the new algorithm is considerably simpler. This result is presented .	We next review the setup and the outline how this improvement is achieved.
    
    \paragraph{Model Description} The \congest model is a synchronous message passing model for networked computation. The communication network is abstracted as a graph $G=(V, E)$, with $n$ vertices. There is one processor on each node of the network, each with unique identifier from $\set{1, \dots, \poly(n)}$. Initially, the network topology is not known to the nodes, except for some global parameters such as as constant factor upper bound on the number $n$ of nodes. The computation is performed in synchronous rounds, each round consists of the phase of (possibly unlimited) local computation and the phase of communication. In the communication phase, each processor is allowed to send a message of size $\bigO(\log n)$ to each of its neighbors. 
    
    For the \MC problem, the question is to identify the edges with the smallest cardinality whose removal disconnects the network. The output will be in a distributed format, meaning that each processor/node should know its own edges in the identified cut. 
    
		\paragraph{The algorithm of Daga et al.} The algorithm of Daga et al.~\cite{daga2019} is actually a mixture of a new algorithm designed for graphs with large edge connectivity, run along an older algorithm of Nanongkai and Su~\cite{Nanongkai2014-distributed-cut}, which is well-suited for graphs with small edge connectivity and runs in $\tilde{O}((D+\sqrt{n})\lambda^4)$. The new algorithm runs in $O(n/\delta^{1/88}) \in O(n/\lambda^{1/88})$ rounds, hence the faster of the two algorithms gives the round complexity of $\tilde{O} (n^{1-1/353}D^{1/353}+n^{1-1/706})$. A key component in this $O(n/\lambda^{1/88})$ rounds algorithm is a procedure based on some variant of expander decompositions that carefully determines parts of the graph that can be contracted, while preserves any non-singleton minimum cut. To be more precise, they provide an algorithm that in sublinear time of $\tilde{O}(n/\delta^{1/88})$ rounds, finds a number $k=O(n/\delta^{44})$ of disjoint connected subgraphs $V_1$, \dots, $V_k$, with a total induced diameter of $\sum_{i=1}^k D(V_i)=O(n/\delta^{1/40})$, such that contracting these subgraphs preserves any non-singleton minimum cut. They then explain an algorithm that in $\tilde{O}(D+ k + \sum_{i=1}^k D(V_i)=O(n/\delta^{1/40})$ extra rounds, identifies the minimum cut of the contracted graph.
		
		\paragraph{Our Improvement} Our proposal is to replace the clever and somewhat sophisticated procedure of Daga et al.~\cite{daga2019} for finding these contractions with just a random $2$-out contractions. From \Cref{l:node_compression_number_of_components}, we know that the number $k$ of components is $\tilde{O}(n/\delta)$. In \Cref{lem:diameter}, we prove an additional nice property: that the components of $2$-out has a summation of diameters $\sum_{i=1}^k D(V_i)=O(n/\delta)$. Hence, just plugging $2$-out in the framework of Daga et al.~\cite{daga2019}, we can improve their $O(n/\delta^{1/88}) \in O(n/\lambda^{1/88})$ round algorithm to run in just $O(n/\delta) = O(n/\lambda)$ rounds. Combining this again with the $\tilde{O}((D+\sqrt{n})\lambda^4)$-round algorithm of Nanongkai and Su~\cite{Nanongkai2014-distributed-cut} gives us our claimed $\tilde{O}(n^{0.8} D^{0.2} + n^{0.9})$ round complexity\footnote{We believe that by plugging in the result of our contraction --- concretely, $\tilde{O}(n/\delta)$ components, each of diameter $O(\log^2 n)$, such that contracting them preserves any particular non-trivial minimum-cut with a constant probability, as we shown in \Cref{remark:3outDiam} --- into the algorithm of Nanongkai and Su~\cite{Nanongkai2014-distributed-cut}, we can improve also the complexity of their algorithm to $\tilde{O}((D+\sqrt{n/\lambda})\lambda^4)$. We have discussed this with one of the authors Hsin-Hao Su. However, this claim should be taken with a grain of salt until all the details are written. If correct, that would lead to a further improved overall round complexity of $\tilde{O}(n^{0.8} D^{0.2} + n^{8/9})$. Moreover, it would give a $(1+\eps)$-approximation of minimum cut with round complexity $\tilde{O}((D+\sqrt{n/\lambda}))$, for any constant $\eps>0$, hence matching the lower bound of Ghaffari and Kuhn~\cite{distributed-cut} for any distributed approximation algorithm on simple graphs.}.

  \subsection{\MPC model}
	\label{subsec:MPC}
    
    In \Cref{s:mc_mpc}, we give an algorithm that computes a minimum cut of a simple graph in $\bigO(1)$ round of the \MPC model, which proves \Cref{t:mc_mpc}. The algorithm is based on our contraction process, after which the input graph is sufficiently small and can be gathered in the memory of a single machine, which can compute the result locally. In \Cref{s:mc_mpc}, we focus on efficient implementation of the contraction process.
    \begin{lemma}
      Given a simple input graph $G$, with $n$ vertices, $m$ edges, and minimum degree $\delta$ it is possible to compute a graph with $\bigO\paren{\frac{n}{\delta}}$ vertices and $\bigO(n)$ edges which for a constant $\eps \in [0,1)$ preserves all non singleton $(2-\eps)$-minimum cuts, \whp, in $\bigO(1)$ rounds of the \MPC model, with $\bigO(m) + \bigO(n\log^3 n)$ global memory, and $\bigO(n)$ memory limit for a single machine. 
    \end{lemma}
    On the top level, the algorithm executes $\Theta(\log n)$ contraction processes in parallel, each reducing the number of edges to $\bigO(n)$ and preserving a fixed $(2-\eps)$-minimum cut with a constant probability, and merges the result using an approach proposed in \cref{subsec:succAmp}. Conceptually, a version of the contraction process mentioned in\cref{t:main_theorem} tailored to the \MPC model consist of computing a random $2$-out contraction followed by contracting a set of edges sampled independently with probability $\frac{1}{2\delta}$. To implement such process, we generate edges to contract beforehand and execute both stages of the contraction process at once. The main obstacle we face is running $\Theta(\log n)$ contraction processes in parallel within $\bigO(m) + \tilde\bigO(n)$ global memory.

To do so, we run $\Theta(\log n)$ instances of Connected Components algorithm ~\cite{JurdzinskiNowicki2018} in parallel, and recover all edges of $G$ that survived any contraction process, in $\bigO(1)$ rounds, with global communication $\bigO(m) + \bigO(r \log n)$, where $r$ is the total number of recovered edges.

  \subsection{\PRAM model}
  
    In \Cref{s:mc_pram}, we give an implementation of the \MC algorithm for simple graphs in the \CPRAM model that proves \Cref{t:mc_pram}. As in the case of sequential model, the algorithm consists of the contraction process followed by application of the state of the art algorithm for general graph. In \Cref{s:mc_pram} we focus on providing an implementation of the contraction process.
    \begin{l:mc_pram}
      Given a simple input graph $G$, with $n$ vertices, $m$ edges, and minimum degree $\delta$ it is possible to execute $\Theta(\log n)$ contraction processes, each resulting with a graph with $\bigO\paren{\frac{n}{\delta}}$ vertices and $\bigO(n)$ edges, on a \CPRAM machine, with $\bigO(m \log n)$ work and depth $\bigO(\log n)$. With high probability one of computed contractions preserves a fixed non singleton minimum cut of $G$.
    \end{l:mc_pram}
	The top level implementation in the \CPRAM model is the same as for the \MPC model: we use a variant of the contraction process based on contracting $2$-out random subgraph followed by contraction of the set of edges to which we include each edge of the input graph with probability $\frac{1}{2\delta}$. The implementation is rather straightforward and boils down to solving several Connected Components problems in parallel, which we do with $\bigO(m \log n)$ work and $\bigO(\log n)$ depth \cite{Pettie2004_MSF_EREW_PRAM}. Then, we use an approach proposed in \cref{subsec:succAmp} to compute an $\bigO(n)$ edge graph that preserves all non singleton $(2-\eps)$-minimum cuts, and find a minimum cut of computed graph with state of the art algorithm for multigraphs in $\bigO(n \log^4 n)$ work and $\bigO(\log^3 n)$ depth ~\cite{Geissmann2018_PRAM_MC}.

\section{\MC in the sequential model}
\label{sec:Seq}

\subsection{An $O(m\log n)$ Algorithm for \MC}
\label{subsec:firstSeq}
To find the edge connectivity of $G$ in $O(m\log n)$ time,
we essentially just apply Gabow's algorithm~\cite{Gabow1991} to the contracted multigraph $\widehat G$
from \Cref{thm:main}. We note that within this time bound, using
another algorithm of Gabow~\cite{Gab16}, we can find the
cactus representation described in  \cite{DKL76} which elegantly represents all
min-cuts of $G$. In particular, we use Gabow's
result in the form described in the lemma below with $M=\widehat G$,
$m_H=O(n)$, $n_H=O(n/\delta)$, and $k=\delta$, yielding a running time
of $O(k m_H\log n_H)=O(n\delta\log n)=O(m\log n)$.
\begin{lemma}[{\cite{Gabow1991}}]\label{lem:Gab}
Given a multigraph $M$ with $m_M$ edges and $n_M$ nodes,
and a number $k$, in  $O(k m_M\log n_M)$ time, we can decide if
the edge connectivity $\lambda_M$ is below $k$, and if so, decide $\lambda_M$ 
it exactly.
\end{lemma}
\begin{proof}
Gabow~\cite{Gabow1991} states the running time as
$O(m_M+n_M\lambda_M^2\log n_M))$ where $\lambda_M$ is the edge
connectivity of $M$. Since $\lambda_M$ is no bigger than the
average degree, his bound is bounded by $O(m_M\lambda_M\log
n_M))$. Moreover, using Matula's linear-time approximation algorithm~\cite{Matula93}, we can decide
the edge connectivity of $\lambda_M$ within a factor 3, hence either
decide that it is bigger than $k$, which we report, or that it is at
most $3k$, implying that Gabow's algorithm runs in 
$O(m_M\lambda_M\log n_M))= O(m_M k\log n_M)$ time.
\end{proof}
We are now ready to show how we find the edge connectivity of $G$ in
$O(m\log n)$ time.
\begin{theorem}\label{thm:Gabow} Let $G$ be a simple graph with $m$ edges and $n$ nodes.
Then, in $O(m \log n)$ time, whp, we can find the
edge connectivity of $G$ as well as a cactus representation of all its minimum
cuts.
\end{theorem}
\begin{proof} As described above, we apply \Cref{lem:Gab} to $\widehat G$
with $k=\delta$ in $O(n\delta\log n)=O(m\log n)$ time. If the edge connectivity of
$\widehat G$ is above $\delta$, the edge connectivity of $G$ is $\delta$,
and all min-cuts are trivial. Moreover, we can easily build a min-cut cactus representation for this case: a star graph with two
parallel edges to all min-degree vertices in $G$.  Otherwise the edge
connectivity of $G$ is the minimum of $\delta$ and that of $\widehat
G$. In this case, we can apply \Cref{lem:Gab} to $\widehat G$, which runs in $O(n\delta\log n)=O(m\log n)$ time. Furthermore, we can also apply Gabow's cactus
algorithm~\cite{Gab16} to $\widehat G$ in $O(m\log n)$ time. In
\cite{KT19:edge-conn}, it is detailed how we convert the cactus of
$\widehat G$ to one of $G$.
\end{proof}

\paragraph{Corollary for Dynamic Graphs} Goranci et al.~\cite{GHT18:inc-edgeconn} have shown how the 
edge contraction from \cite{KT19:edge-conn} preserving all $3/2$-small cuts 
can be used in an incremental (edge insertions only) algorithm
maintaining edge connectivity. By plugging \Cref{thm:main} instead
of the algorithm from \cite{KT19:edge-conn} in their framework, we get the following corollary.
\begin{corollary} We can maintain the exact edge connectivity of an incremental
dynamic simple graph, whp, in $O(\log n)$ amortized time per edge insertion.
\end{corollary}
We note the above bound holds against an adaptive user where future
updates may depend on answers to previous queries.

\subsection{Faster Contractions and an $O(m+n\log^3 n)$ Algorithm for \MC}
\label{subsec:secondSeq}

We now present a faster contraction algorithm for dense graphs. Later in \Cref{thm:fasterCut}, we explain how this leads to an $O(m+n\log^3 n)$ algorithm for edge connectivity.

\subsubsection{Faster Contraction via Data Structures}
\begin{theorem}\label{thm:main-dense} Let $G$ be a simple graph with $m$ edges, $n$ nodes, and
min-degree $\delta$. We have a randomized algorithm, running in $O(m+n
\log(n)\alpha(n,n))$ time, contracting edges of $G$ so that the resulting
multigraph $\widehat G$ has $O(n)$ edges and $O(n/\delta)$ nodes, and 
preserves all non-trivial $(2-\eps)$-small cuts of $G$ whp.
\end{theorem}

\paragraph{Note} In the above statement, $\alpha$ is the extremely slow-growing inverse Ackermann
function that Tarjan~\cite{T75} used to bound the complexity of the
union-find data structure. He showed that union-find with $u$ unions and $f$ finds 
over $s$ elements, initialized as a singleton sets, 
can be supported in $O(s+\alpha(f,u)f)$ total time. Here 
$\alpha$ is decreasing in $\lceil f/u\rceil$. We have 
$\alpha(f,u)=\alpha(u,u)$ if $f\leq u$ and $\alpha(f,u)=O(1)$ if,
say, $f\geq u\log\log\log u$. In general, we will use union-find to grow certain forests efficiently, in the following classic way.  We
are growing a forest $F$, and the union-find sets are the node sets
spanned by the trees in $F$. Initially, the forest has no
edges, and the nodes are singleton set. If we get an edge $(u,v)$,
we can use finds on $u$ and $v$ to check if they are spanned
by a tree in $F$. If not, we can add $(u,v)$ to the forest.

\medskip
\paragraph{Outline} Our main tool to prove \Cref{thm:main-dense} is the following 
on-line data structure version of the contraction process in \Cref{t:main_theorem}. After presenting this lemma, we use it to prove \Cref{thm:main-dense}.
\begin{lemma}\label{lem:3-out+uf} 
Let $G$ be a given a simple graph with $n$ nodes, $m$ edges, and minimum
degree $\delta$. We will construct a randomized $O(n)$ space data
structure $\widehat D$ that we feed edges from $G$ in any order, but
without repetitions. When given an edge, the data structure will
answer ``preserve'' or ``contract''.  The order we feed edges to
$\widehat D$ may adaptively depend on previous answers made be
$\widehat D$.  The data structure provides the following two guarantees: 
\begin{itemize}
\item The data structure $\widehat D$ answers preserve to at most $O(n)$ edges. 
\item Let $e^*$ be any edge of $G$ belonging to some $(2-\eps)$-small
cut $C^*$ of $G$. Then, with probability at least $p_\eps/2$, 
$\widehat D$ will answer preserve if queried on $e^*$\;\footnote{As a subtle
point, note that we are \emph{not} claiming that we with constant probability
simultaneously will preserve all edges from $C^*$.}. Here $p_\eps$ is
the constant probability from \Cref{t:node_compression}.
\end{itemize}
Finally, if the number of edges fed to $\widehat D$ is $f$, then
the total time spent by the data structure is $O(n+\alpha(f,n)f)$.
\end{lemma}

\begin{proof}
The proof has steps that mimic that process of \Cref{t:main_theorem}, but in a more efficient way and as a data structure. In particular, we will have a part for $2$-out contraction, and a more elaborate part that mimics the effect of sparse-certificates. Next, we present these two parts.
 
\paragraph{First part} Given the graph $G$, first we make a 2-out sample $S$. We color all
the components of the graph with edge set $S$ in $O(n)$ time so that different components have different colors, i.e., we identify the components.  For
each vertex $v$, we store its component color $c(v)$. These component
colors are the vertices in $G/S$. An edge $(u,v)$ from $G$ corresponds
to an edge $(c(u),c(v))$ in $G/S$. Here edges preserve their edge
identifies, so if another edge $(u',v')$ has $c(u)=c(u')$ and
$c(v)=c(v')$, then $(c(u),c(v))$ and $(c(u'),c(v'))$ are viewed as
distinct parallel edges.

By \Cref{t:node_compression}, whp, $G/S$ has $n_{G/S}=O(n/\delta)$ vertices. If
this is not the case, we create a trivial data structure contract 
to all edges, so assume $n_{G/S}=O(n/\delta)$ vertices. By
\Cref{t:node_compression}, the probability that $C^*$ is preserved in $G/S$
is at least $p_\eps$. Assume below that $C^*$ is preserved in $G/S$.

\paragraph{Second part} This part intends to mimic the effect of sparse certificates---intuitively (though, not formally) similar to growing $4\delta$ maximal forests, one after another.
In particular, we initialize $\ell=4\delta$ union-find data structures to grow
edge-disjoint forests $F_1,\ldots,F_\ell$ over the vertices in $G/S$.
Initially, there are no edges in the forests. From a union-find
perspective, we can think of it as if we have $\ell$ disjoint
copies of the nodes in $G/S$, so to ask if $c(u)$ is connected to $c(v)$
in $F_i$, we ask if $c(u)_i$ is in the same set as $c(v)_i$.

When given an edge $(u,v)$ from $G$, we can ask if it is connected in some
$F_i$ in the sense that $c(u)$ and $c(v)$ belong to the same tree in
$F_i$.  If not, we can $(c(u),c(v))$ to $F_i$. Note that if $(u,v)$
got contracted in $G/S$, then $c(u)=c(v)$, and then $c(u)$ and $c(v)$
are trivially connected in every forest $F_i$.
We will follow the rule that each edge $(u,v)$ from $G$ may be
added as an edge $(c(u),c(v))$ to a single forest
$F_i$. This way the forests remain edge-disjoint,
so if $c(u)$ and $c(v)$ are
connected in $k$ different forests $F_i$, then $c(u)$ and $c(v)$ are 
$k$-connected by the edges added to all the forests. This
implies that $u$ and $v$ must be $k+1$ connected in $G/S$.

Now, consider our edge $e^*=(u^*,v^*)$ from our
$(2-\eps)$-small cut $C^*$ of $G$ which we assumed was preserved in
$G/S$. Then $c(u^*)$ and $c(v^*)$ are at most $|C|<2\delta$ connected in
$G/S$, so $c(u^*)$ and $c(v^*)$ are connected in less than half
of the $4\delta$ forests $F_i$.
This leads to the following randomized algorithm to handle a
new edge $(u,v)$ from $G$. We pick a uniformly random index
$i\in\{1,\ldots,\ell\}$ and ask if $c(u)$ and $c(v)$ are connected in
$F_i$. If not, we add $(c(u),c(v))$ to $F_i$ and answer
``preserve''. Otherwise, we answer ``contract''. The latter would be a
mistake on $e^*$, but assuming that $G/S$ preserved $C^*$, we know
that $c(u^*)$ and $c(v^*)$ are connected in less than half the
forests, and then the probability of a false ``contract'' is bounded by
$1/2$.

\medskip
\paragraph{Overall error probability} 
For the overall probability on correctly answering ``preserve'' on $e^*$, we
first want the good event that $S$ to not intersect $C^*$. 
By \Cref{t:node_compression}, 
this first good event happens with probability at least $p_e$.
Conditioned on the first good event, we want
our the random forest $F_i$ to answer ``preserve'' on $e^*$, which
happened with probability at least $1/2$, so the overall probability 
that we answer ``preserve'' on $e^*$ is at least $p_\eps/2$
as desired.

\medskip
\paragraph{The number of preserved edges}
The edges we add to $F_i$ form a forest over the $O(n/\delta)$ nodes
in $G/S$, so we can only add $O(n/\delta)$ edges to each of
the $\ell=4\delta$ different $F_i$. Thus we conclude that
there are at most $O(n)$ edges that we add to the edge-disjoint forests
$F_i$. 

\medskip
\paragraph{Time complexity}
Each time we get an edge, we check connectivity of $s=O(1)$
forest, so we make $O(f)$ find operations. We have $\ell$
copies of each node in $G/S$, so the total number of
elements is $\ell n_{G/S}=O(n)$. We conclude
that the total time spent by our
data structure is $O(n+\alpha(f,n)f)$.
\end{proof}

Having this helper data structure version of a single contraction process, we are now ready to prove \Cref{thm:main-dense}. That is, we present an $O(m+n
\log(n)\alpha(n,n))$ time contraction down to $O(n)$ edges and $O(n/\delta)$, while preserving all non-singleton $(2-\epsilon)$ minimum cuts, with high probability.
\begin{proof} [Proof of \Cref{thm:main-dense}]
We are going to use the data structure from \Cref{lem:3-out+uf} 
in much the same way as we used contracted graph from 
\Cref{t:main_theorem}. Concretely, we want to amplify the success using certain repetition and voting rules so that we preserve all non-singleton $(2-\epsilon)$ minimum cuts, with high probability. Below we present a direct translation, and later we show how to tune it.

\medskip
\paragraph{The procedure} Let $p'_\eps=p_\eps/2$ be the error probability from \Cref{lem:3-out+uf}.  We apply the lemma $q'=O((\log n)/p'_\eps)=O(\log
n)$ times, initializing independent data structures $\widehat
D_1,\ldots,\widehat D_{q'}$.

Consider any edge $e^*$ belonging to some $(2-\eps)$-small cut $C^*$ of
$G$.  Each $\widehat D_i$ preserves $e^*$ with probability at least
$p'_\eps$, so the expected number of $\widehat D_i$ preserving $e^*$
is at least $\mu'=p'_\eps q'$.  By Chernoff, the probability that less than
$r'=\mu'/2=p'_\eps q'/2$ of the $\widehat S_i$ want to preserve $e^*$
is bounded by $\exp(-\mu'/8)=\exp(-p'_\eps q')$. For any given
$\gamma$, this is $O(n^{-\gamma})$ for $q'\geq 8(\ln n)\gamma/p'_\eps=16(\ln
n)\gamma/p_\eps$.

We can now take each edge $e$ in $G$, and query all the $\widehat
S_i$ counting how many answer preserve. If this number is less than
$r'$, then we contract $e$. Otherwise we say that $e$ was \emph{voted
  preserved}. However $e$ could still be lost as a self-loop due to
other contractions.

\medskip
\paragraph{The number of preserved edges}
Since each $\widehat S_i$ answers preserve
for $O(n)$ edges, the total number of preserve answers
from all the $\widehat S_i$ 
is $q'\cdot O(n)$. The number of edges that are preserved $r'$ times
is therefore $q\cdot O(n)/r' =O(\gamma n/p_\eps)=O(n)$. All
other edges are contracted, so the resulting graph $\widehat G$
ends up with $O(n)$ edges, as desired for \Cref{thm:main-dense}.

\medskip
\paragraph{We preserving all small cuts}
As described above, our edge $e^*$ from some $(2-\eps)$-small cut was
voted preserved with probability $1-O(n^{-\gamma})$. This implies that
with probability $1-O(n^{-\gamma'})$ with $\gamma'=\gamma-2$, every
edge $e$ belonging to any $(2-\eps)$-small cut is preserved. In
particular, given any $(2-\eps)$-small cut $C$, we get that all edges
in $C$ are voted preserved meaning that none of them are contracted
directly. Because $C$ is a cut this implies that all of $C$ survives
the contractions, that is, none of the edges in $C$ can be lost as
self-loops due to other contractions. Thus we conclude that, whp,
$\widehat G$ preserves all $(2-\eps)$-small cuts $C$ in $G$.

\medskip
\paragraph{An issue with the complexity, and fixing it}
Unfortunately, our total run time is still bad because for every
edge in $G$, we query all $q$ data structures $\widehat D_i$. However, as
our last trick, we maintain a union-find data structure telling
which vertices in $G$ that have already been identified due to
previous contractions. We now run through the edges of $G$ as
before, but if we get to an edge $(u,v)$ where $u$ and $v$
have already been identified, then we skip the edge since contracting
it would have no effect. All the above analysis on the properties 
of $\widehat G$ is still valid.

With the above change, checking if end-points have already been
identified, we claim that we can make at most $O(n)$ queries to the
data structures. We already saw that we could only vote preserve $O(n)$
times. However, if the data structures instead vote to contract
$(u,v)$, then this is a real contraction reducing the number of
vertices in $\widehat G$, so this can happen at most $n-1$ times.

Every time we query the data structures, we query all $q=O(\log n)$ of them.
In each we have $f=O(n)$ queries, so the time spent in each is 
$O(\alpha(n,n)n)$, adding up to a total time of 
$O(n(\log n)\alpha(n,n))$. In addition, the top-level union-find data
structure over the contracted vertices in $\widehat G$ uses
$O(n+\alpha(m,n)m)\ll O(m+n\log\log\log n)$ time, so our total
time bound is $O(m+n(\log n)\alpha(n,n))$. This completes the
proof of \Cref{thm:main-dense}.
\end{proof}

\subsubsection{Faster Minimum Cut in Dense Graphs}
We can now complete our $O(m+n\log^3 n)$ algorithm.
\begin{theorem}\label{thm:fasterCut} We can find the edge connectivity and some min-cut
of a simple graph $G$ with $m$ edges and $n$ nodes in $O(m+n\log^3 n)$ whp.
We can also find the cactus representation of all min-cuts of $G$ in
$O(m+n\log^{O(1)} n)$ time.
\end{theorem}
\begin{proof}
Much like we used Gabow's algorithm's~\cite{Gabow1991,Gab16} 
on the contracted graph in \Cref{thm:main} to prove \Cref{thm:Gabow}, we now apply Karger's~\cite{Karger2000} edge connectivity
algorithm to the contracted graph in \Cref{thm:main-dense}, 
and his algorithm with Panigrahi~\cite{DBLP:conf/soda/KargerP09} to
get the cactus representation.
\end{proof}
Note that because we only spend $O(m+n(\log n)\alpha(n,n))$ time on
constructing the contracted graphs, we would instantly get better
results, if somebody found an improvement to Karger's algorithm~\cite{Karger2000}.  However, due to parallel edge resulting from
contractions, it has to be an algorithm working for general graphs, so
we cannot, e.g., use the recent algorithm of Henzinger et
al.~\cite{DBLP:conf/soda/HenzingerRW17}.

\section{\MC in the \congest model}
\label{sec:mc_congest}

To obtain the $\tilde \bigO\paren{n^{0.8}D^{0.2} + n^{0.9}}$ bound claimed by \cref{t:mc_congest}, we use the general approach proposed by Daga et al.~\cite{daga2019}, except that we replace a key component of their work (a contraction based on expander decompositions that preserves nontrivial minimum cuts) with random $2$-out contractions. This leads to a substantial simplification and time complexity improvement. 

    In their work~\cite{daga2019},  Daga et al. propose an algorithm that given a partition of vertices into disjoint connected sets $P = \set{V_1 \cup V_2 \dots \cup V_k}$, such that the edges on non singleton minimum cuts are only between the sets of vertices --- i.e., contracting these sets preserves non-trivial minimum cut --- finds the minimum cut of the graph $G/P$ in time $\tilde \bigO(D(V) + k + \sum_{i=1}^k D(V_i))$. Here, $G/P$ is a graph in which we contract all sets $V_1, V_2 \dots V_k$ into vertices and $D(S)$ is a diameter of a graph induced by $S$. We will make use of this algorithm, in a black-box fashion. But let us briefly discuss how  Daga et al. used it: 
    They provide an algorithm based on \emph{expander decompositions} that in sublinear time of $\tilde{O}(n/\delta^{1/88})$, finds such a partition with sublinear $k=O(n/\delta^{44})$ and $\sum_{i=1}^k D(V_i)=O(n/\delta^{1/40})$. This leads to an $\tilde{O}(n/\delta^{1/88})$ rounds algorithm, which is sublinear for graph with sufficiently large minimum cut. They then obtain their round complexity by combining this algorithm for graphs with large minimum cuts with the algorithm by Nanongkai and Su~\cite{Nanongkai2014-distributed-cut} for graphs with small minimum cut.
		
		    To get a faster algorithm, we observe that random $2$-out contractions provide a vastly simpler and also more efficient way of contracting the graph into fewer nodes (where contractions have small diameter on average), while preserving non-trivial minimum cuts, as desired in the algorithm of Daga et al.~\cite{daga2019}. As a result, by plugging in $2$-out contractions in the approach of Daga et al., we obtain an $\tilde{O}(n/\delta)$ round algorithm. This is a substantial improvement on the $\tilde{O}(n/\delta^{1/88})$ round algorithm of Daga et al. (notice that these algorithms will be applied for graphs with large minimum cut and thus large $\delta$). Again, combining this with the algorithm of Nanongkai and Su~\cite{Nanongkai2014-distributed-cut} for graphs with small minimum cut gives the final round complexity.

In the next two subsections, we first analyze the diameter of the components of a $2$-out, and then plug this property into the framework of Daga et al.~\cite{daga2019} to obtain our faster distributed edge connectivity algorithm.		

 \subsection{Diameter of the components of $2$-out subgraph}
We now upper bound the total diameter of the connected components of the $2$-out subgraph.

\begin{lemma}\label{lem:diameter}  With high probability, the sum of diameters of the components of a random 2-out is $O(n(\log \delta)/\delta)$.
\end{lemma}
\begin{proof}
Since the total diameter is trivially bounded by $n$, we can assume
$\delta=\omega(1)$.

For any given vertex $v$, we let $B_i(v)$ denote the ball of radius $i$ around
$v$ in the $2$-out sampled subgraph. Here $B_0(v)=\{v\}$. We call
$L_i(v)=B_i(v)\setminus B_{i-1}(v)$ the $i$th level around $v$. We
now grow the ball $B(v)$ around $v$, one level at the time. Suppose we have
already grown $B_i(v)$. To add the next level, we take the vertices
$u\in L_i (v)$, one at the time, and generate its two out-edges
$(v,w)$, one at the time. The edge ``gets out'' if leads to a new
vertex $w$ which is not already in $B(v)$. In this case $w$ is added
to $B(v)$; otherwise the sampled edge ``gets stuck''. Since the
next sampled edge is sampled uniformly from the at least $\delta$ edges
leaving $u$, the probability that it gets stuck is at most 
$|B(v)|/\delta$.

Key to our proof, we show the
following lemma:
\begin{lemma}\label{lem:ball} For any given vertex $v$, we have\footnote{$\lg=\log_2$.}
\[\Pr[|B_{2\lg (\delta/200)}(v)|\leq\delta/200]\leq (200/\delta)^2.\]
\end{lemma}
\begin{proof}
Set $s=\delta/200=\omega(1)$. We want to show that $\Pr[B_{2\lg
    s}(v)\leq s]\leq 1/s^2$. We say the ball $B_{2i}(v)$ is successful if $|B_{2i}(v)|\geq s$
or $|L_{2j}(v)|\geq 2 |L_{2(j-1)}(v)|$
for $j=1,\ldots,i$. For $i\leq \lg s$, this implies $|L_{2i}(v)|\geq 2^i$.
It also implies that the total number of even level vertices in $B_{2i}(v)$
is less than $2|L_{2i}(v)|$. 

Assume that the ball $B_{2i}(v)$ is successful and $|B_{2i}(v)|<s$.
Set $b=|L_{2i}(v)|$. We want to bound the probability that 
$|L_{2(i+1)}(v)|<2b$. Consider the $2$ level
boundary growth from the vertices in $L_{2i}(v)$. The growth from 
$u\in L_{2i}(v)$ can lead to at most $6$ new vertices, and this holds whenever we grow from an even
level vertex, so the total number of vertices reached in $B_{2(i+1)}(v)$ is
at most $6$ times the number of even level vertices in $B_{2i}(v)$,
hence at most $12 b$. This means that as we grow $B(v)$ from $B_{2i}(v)$ to $B_{2(i+1)}(v)$,
the probability that any sampled edge gets stuck is at most $12b/\delta$.

If no sampled edge got stuck during the 2-level growth from $L_{2i}(v)$, 
then we would get $|L_{2i+2}(v)|= 4b$. Each sampled edge
getting stuck, can reduce this number by at most 2, so to get 
$|L_{2(i+1)}(v)|<2b$, we need at least $b+1$ sampled edges to get stuck,
and this is out of at most $6b$ sampled edges.
The probability of this error event is bounded by 
\[p_b={6b\choose b+1}(12b/\delta)^{b+1}<(72eb/\delta)^{b+1}.\]
For $2\leq b\leq \delta/200$, we have $p_b=O(1/\delta^3)$. We
do the two level growth at most $\lg s$ times, so
the probability that we fail to get to size $s$ is
$(72e/\delta)^{2}+O((\log\delta)/\delta^3)<(200/\delta)^2$.
\end{proof}
Around every vertex $v$, we consider the ball $B(v)=B_{r}(v)$ with
radius $r=2\lg (\delta/200)$ as in Lemma \ref{lem:ball}. We say
  that two vertices are ball neighbors if their balls intersect. Then
  a ball path of length $\ell$ is a sequence of vertices
  $v_0,\ldots,v_\ell$ where $B(v_i)$ and $B(v_{i+1})$ intersect. This
  implies that there is a regular path from $v_0$ to $v_\ell$, passing through
  $v_1,\ldots,v_{\ell-1}$, of length at most
  $2r\ell$. We define ball distance and ball diameter of
  a component in our 2-out subgraph accordingly. To prove that the diameter sum is
  $O(n(\log\delta)/\delta)$, it suffices to prove that the ball
  diameter sum is $O(n/\delta)$.

Consider some component $A$ of our 2-out subgraph. Suppose $A$ has
ball diameter $\Delta$. This means that there are two vertices $v$ and
$w$ such that the shortest ball path between them is a ball path
$v=v_0,\ldots,v_\Delta=w$ of length $\Delta$. Because this is a
shortest ball path, we know that the $\lceil \Delta/2\rceil$ balls of the
even vertices $B(v_0),B(v_2),\ldots ,B(v_{\lceil \Delta/2\rceil})$ are
all disjoint. It follows that if the ball diameter sum is $\Delta^*$,
then our 2-out subgraph has at least $\Delta^*/2$ vertices with disjoint balls. 
Thus
the theorem follows if we can prove that, whp, there can only be only
$O(n/\delta)$ vertices with disjoint balls.

Let $k=c n/\delta$ where $c$ is some large constant which is at least
$400$. We can pick
a set $U$ of $k$ vertices in ${n \choose k}< (en/k)^k=(e\delta/C)^k$
ways. For any such set $U$, we will show that the probability that the
vertices in $U$ have disjoint balls is very small.

We take the vertices $v\in U$, one at the time, and grow the ball 
$B(v)=B_r(v)$. As long as $B(v)$ has not intersected
any previous ball, the growth with new edges for $B(v)$ is completely 
independent of the samples done growing balls from previously considered 
vertices in $U$. We say $B(v)$ fails if we get
$|B(v)|\leq\delta/200$ while $B(v)$ does
not intersect any previous ball. By Lemma \ref{lem:ball}, the
failure probability is bounded by $p=(200/\delta)^2$, and this 
is no matter how previous balls were grown. 

We can have at most $200 n/\delta\leq k/2$ non-intersecting balls
of size $\delta/200$, so to stay disjoint, we must have at least
$k/2$ failing balls from the given set $U$. However, we only expect $pk$ 
failing balls, so by Chernoff, the probability of getting
$k/2$ failing balls, is bounded by
\[(e/(1/(p/2))^{k}=(2e(200/\delta)^2)^k.\]
This then bounds the probability that the balls from $U$ are all disjoint.
Union bounding over the less than $(e\delta/c)^k$ choices for the set $U$,
we conclude that the probability of getting any $k$ disjoint balls in our
$2$-out subgraph is at most
\[(e\delta/c)^k(2e(200/\delta)^2)^k=(80000e^2/(\delta c)^k.\]
With $c\geq 80000e^2$, this is bounded by
$1/\delta^k=1/\delta^{cn/\delta}$.  This bound is maximized for
$\delta=n$, so our probability of getting $k$ disjoint balls is
bounded by $n^{-c}$. Therefore, whp, we get at most
$O(n/\delta)$ disjoint balls of radius $r=2\lg(\delta/200)$. The ball
diameter sum was at most twice as big as the number of disjoint balls,
and the diameter sum was only $2r$ times bigger than the ball diameter
sum. Hence, whp, diameter sum of our 2-out subgraph is $O(n(\log\delta)/\delta)$.
\end{proof}

The above theorem is tight in the sense that if a graph consists of $n/\delta$ disjoint cliques of size
$\delta$, then whp, the diameter sum of a 2-out subgraph is $\Theta(n
  (\log \delta)/\delta)$.

\begin{remark}
\label{remark:3outDiam}
 We can choose a subset of the edges of the $2$-out so that the spanning subgraph defined by them has $\bigO((n/\delta)\log n)$ components, each with diameter $\bigO(\log n)$.
\end{remark}
\begin{proof}
Our goal is to pick $O((n/\delta)\log n)$ centers. Each
vertex use its nearest center, which should be
at distance $O(\log n)$. All we need to keep are shortest
path trees from the centers to the vertices that use them.
If $\delta=O(\log n)$, we can just pick all vertices as centers,
so we may assume that $\delta\geq C\log n$ for an arbitrarily large
constant $C$.

Basic idea is as follows. In the 2-out subgraph, we will show that
most vertices have $\Theta(\delta)$ vertices at distance $O(\log n)$,
and all such vertices are served, whp, if we pick $\Theta((n/\delta) \log n)$
random centers. The remaining vertices will be served
from $O((n/\delta) \log n)$ special centers.

First use a 1-out sample $S_1$. For any $x$, if we start
from a vertex $v$, and follow the 1-out edges, the probability
that we do not reach $c$ vertices is less than $x^2/\delta$. Hence,
as in the proof of \cref{l:node_compression_number_of_components}, by Chernoff we conclude that
only $O(n/\delta+\log n)$ vertices get connected to less than $5$
vertices. For each component of size less than $5$ in $S_1$, we
pick one as a special center, and we call this a special component.

We now consider the components of $S_1$ of size at least 5. 
If any has diameter more than 9, we
cut it into components of diameter at most 9 and size at least 5.
We now have a subgraph of $S_1$ where all components are of diameter at 
most 9. We have $O(n/\delta+\log n)$ special components of size at most 3.
The rest are called regular, and they all have size at least 4.
We contract all these components into super nodes that
are called regular or special if the component was called regular or special.
The size of a super node is the number of subsumed original vertices.
Since $S_1$ only exposed one out edge from each vertex, we know that
a super nodes of size $x$ has $x$ unexposed edges. For regular
nodes, $x\geq 5$. 

Now, from all regular nodes, we expose some of their unexposed out
edges, creating a new sample $S_2$ disjoint and independent from our
first 1-out sample $S_1$. More precisely, if the vertex has size $x$,
we expose $\lceil 3x/5 \rceil$ edges for $S_2$.

Similar to the proof of \cref{l:node_compression_number_of_components}, we
now take one regular node at the time, and follow the $S_2$
growth until we either hit a special node, or reach total size
at least $\ell=c\log n$, for some large enough constant $c$. Recall
here that we have $C\log n\leq\delta$ for an arbitrarily large constant
$C$, so we can pick $c=\sqrt C$.

We want to show an $O(n/\delta+\log n)$ bound on the number of
new $S_2$-components that do not reach a special node or reach size $\ell$. As
in the proof of \cref{l:node_compression_number_of_components}, we
note that we just have to consider the probability of failing to reach
the size without hitting any of the previous components or special
super nodes.

Consider the $S_2$-growth from a regular node $v$. Suppose the final
growth from $v$ exposes $x$ edges. Then the total size spanned 
is at most $5x/3$. We can assume that only regular nodes are reached
and they have size at least $5$, so the exposed edges have connected at 
most $x/3$ regular nodes. This means that more than $2x/3$ of the
exposed edges got caught in the sense that they did not
expand to new regular nodes. The probability
that a given exposed edge got caught is less than $(5x/3)/\delta$. The
edges that did get caught can be chosen in less than $2^x$ ways,
and the chance a given choice of at least $2x/3$ edges get caught is 
bounded by $(2x/\delta)^{2x/3}$. Thus, the probability that
we end up exposing exactly $x$ edges is bounded $P_x=2^x ((5x/3)/\delta)^{2x/3}\leq
(5x/\delta)^{2x/3}$. Since we expose at least $3$ edges from
the first regular node $v$, we have $x\geq 3$, and it
is easily checked that $\sum_{x=3}^\ell=O(1/\delta^2)$ for
$\ell=\delta/c$ for our sufficiently large constant $c$.
It follows that all but $O(n/\delta)$ components reach size $\ell=c\lg n$ or
a special node. 

If an $S_2$-component does not reach size $\ell$ and does not have
a special node, we pick a special center in it. This is
only $O(n/\delta)$ new special centers.

We now take all the components of size at least $\ell$, and
partition them into components of size at least $\ell$ and
diameter at most $2\ell$. If one of these
components contain a special center, we are done with it.
Each of the other components is called a kernel. 
A kernel has only regular nodes, and all nodes that are not
covered by some special center are in some kernel.

We now pick out a single kernel $A$. Exposing all remaining out edges in
a last independent sample $S_3$, whp, within distance $O(\log\delta)$,
we will get to special center or $\Theta(\delta)$ vertices.

When growing from the kernel $A$, we ignore the $S_2$ edges
connecting the other kernels. All we consider are the super nodes
contracted from the $S_1$ and the new exposed edges from $S_3$.
We now that a regular node of size $x\geq $ exposed $\lceil 3x/5\rceil$
edges for $S_2$, so it has $x-\lceil 3x/5\geq x/5$ edges
left for $S_3$. It follows that the number of edges that $S_3$ exposes from $A$
is at least $|A|/5$. Thus, we start by exposing at
least $\ell/5=(c/5)\lg n$ vertices from $A$. Starting from $A'=A$,
we grow $A'$. As long as $|A'|\leq \delta/4$,
then each exposed edge has $3/4$ chance of leaving $A'$, so whp, we 
reach $(c/10)\lg n$ new regular nodes outside $A$. We now
proceed in rounds, each time increasing the distance from the
kernel by 1. 

For the regular nodes outside the kernel $A$, we exploit that if it
has size $x\geq $ it exposed $\lceil 3x/5\rceil$ edges for $S_2$, so
it has $x-\lceil 3x/5\geq 2$ edges left for $S_3$. Hence, as we have
reached less than $\delta/4$ vertices and no special nodes, whp, we
double the number of new nodes. Thus, in less than $\lg \delta$
rounds, we reach $\delta/4$ vertices or a special node. Since the
kernel has diameter $O(\log n)$, we conclude, whp, that all vertices
in the kernel at distance $O(\log n)$ from $\delta/4$ vertices or some
special center. This high probability result must hold for all
kernels, hence for all vertices.
\end{proof}	

\subsection{Improved Distributed Algorithm}
   
    \begin{lemma}\label{l:mc_congest}
      Given any simple input graph $G$ with $n$ vertices, $m$ edges, minimum degree $\delta$, and minimum cut size $\lambda$, it is possible to identify its minimum cut in $\tilde \bigO\paren{\frac{n}{\delta}} =\tilde\bigO\paren{\frac{n}{\lambda}}$ rounds of the \congest model, w.h.p. 
    \end{lemma}
    
    \begin{proof}
    
    Consider a $2$-out contraction of the input graph. Since we consider simple graphs, we have that the graph diameter has a small diameter $D(G) = \bigO\paren{\frac{n}{\delta}}$. By \cref{l:node_compression_number_of_components} we have that in the contracted graph we have at most $\bigO\paren{\frac{n}{\delta}}$ vertices. By \cref{lem:diameter} we have that in the $2$-out contraction the sum of diameters of contracted subgraphs is $\tilde \bigO\paren{\frac{n}{\delta}}$. Therefore, we can identify the minimum cut of the graph obtained by a $2$-out contraction in $\tilde \bigO\paren{\frac{n}{\delta} + \frac{n}{\delta} + \frac{n \poly{\log n}}{\delta}} = \tilde \bigO\paren{\frac{n}{\delta}} = \tilde \bigO\paren{\frac{n}{\lambda}}$ rounds, by applying the algorithm of Daga et al.~\cite{daga2019}.
    \end{proof}
    By combining \Cref{l:mc_congest} with the algorithm of \cite{Nanongkai2014-distributed-cut}, which is best suited for graphs with small edge connectivity $\lambda$, we get our round complexity of $\tilde\bigO\paren{n^{0.8} D^{0.2} + n^{0.9}}$, as claimed in \Cref{t:mc_congest}:
\begin{proof}[Proof of \Cref{t:mc_congest}]
		In \Cref{l:mc_congest}, we provided an algorithm with round complexity $\tilde\bigO\paren{\frac{n}{\lambda}}$. Nanongkai and Su~\cite{Nanongkai2014-distributed-cut} gave an algorithm that runs in $\bigO(\lambda^4\log^2 n(D+\sqrt{n}\log^* n)) = \tilde\bigO(\lambda^4(D+\sqrt{n}))$ rounds. We now explain that by running both algorithms and taking the faster of the two, we can obtain an algorithm with complexity $\tilde\bigO\paren{n^{0.8} D^{0.2} + n^{0.9}}$. For graphs with diameter $D \in \bigO(\sqrt{n})$, the second algorithm requires $\tilde\bigO(\sqrt{n} \lambda^4)$ rounds. Taking the minimum of this and our algorithm that runs in $\tilde\bigO\paren{\frac{n}{\lambda}}$ gives an algorithm that runs in $\tilde\bigO\paren{n^{0.9}}$ rounds.
  For graphs with diameter $D \in \Omega(\sqrt{n})$, the algorithm of Nanongkai and Su runs in $\tilde\bigO(D \lambda^4)$ rounds. Taking the minimum of this and our algorithm that runs in $\tilde\bigO\paren{\frac{n}{\lambda}}$ gives an algorithm that runs in $\tilde\bigO(n^{0.8}D^{0.2})$ rounds. Hence, running both algorithms and taking the faster of the two runs in $\tilde\bigO\paren{n^{0.8} D^{0.2} + n^{0.9}}$ time.
	\end{proof}

\section{\MC in the \MPC model} \label{s:mc_mpc}
  	The \MPC model is a model of parallel (or distributed) computing, in which the computation is executed in synchronous rounds, by a set of $M$ machines, each with \emph{local memory} of size $S$. Every rounds consists of the phase of local computation and the phase of communication. In the phase of local computation each machine can execute some, possibly unbounded computation (although we could consider only computation that takes time polynomial in $S$, or even only $\tilde\bigO(S)$ step computations). In the phase of communication, the machines simultaneously exchange $\bigO(\log n)$-bit messages, in a way that each machine is a sender and a receiver of up to $\bigO(S)$ messages.
  
  	The \emph{global memory} is the total amount of memory that is available, i.e. if there are $M$ machines, each with $S$ memory, their global memory is $MS$. Ideally, for an input of size $N$, the values of $M$ and $S$ are chosen in a way that the global memory is $\bigO(N)$. Assuming that the global memory limit is set to be $\bigO(N)$, we have two main quality measures of the algorithms in the \MPC model: the first one is the number of rounds that are required by the algorithm to finish computation, the second is the limit on the local memory of a single machine. 
  
  	For graph problems in the \MPC model, we distinguish two significantly different variants of the \MPC model, depending on relation between the limit on local memory $S$ and the number of vertices of the input graph, which is usually denoted by $n$. More precisely, those variants are $S \in \tilde\Theta(n)$ and $S \in \bigO(n^{1-\eps})$ for some constant $\eps > 0$. In this paper, we focus on the variant in which the limit on the local memory is $\bigO(n)$ words, each of length $\bigO(\log n)$ bits.
  
  	Sometimes, we also consider the algorithms that have global memory limit larger than $N$. It can be achieved in two ways: by setting higher limit on the local memory of a single machine, or by increasing the number of machines. The first variant of this relaxation is stronger -- a single machine with enlarged memory limit can simulate several machines with a smaller limit. Therefore, if we have two \MPC machines, both with the same global memory, the one with larger limit on the local memory can simulate the one with a smaller limit on local memory.
  
	In particular, in this section we give a \MC algorithm for simple graphs in the \MPC model, that proves \Cref{t:mc_mpc}, i.e. the algorithm requires $\bigO(1)$ rounds of computation, uses $\bigO(m + n \log^3 n)$ global memory, while respecting $\bigO(n)$ memory limit on a singe machine. 
  
	In the remaining part of this section we propose an implementation of the contraction process mentioned in \cref{thm:main}: the algorithm requires $\bigO(1)$ rounds, and works with $\bigO(n)$ limit on the memory of a single machine and $\bigO(m + n\log^3 n)$ global memory. Since the contracted graph has only $\bigO(n)$ edges, we can fit the whole contracted graph and the sizes of all singleton cuts in the memory of a single machine, which then can compute the minimum cut of a contracted graph and compare it with all singleton cuts, which proves \cref{t:mc_mpc}. 
  
\subsection{Contraction process in the \MPC model}
	On the top level, we want to follow the reasoning presented in \cref{subsec:succAmp}. The first part is to compute $\Theta(\log n)$ graphs that have $\bigO(n)$ edges and preserve a fixed $(2-\varepsilon)$-minimum cut at least with some probability $p_{\varepsilon}$. Then, we use a voting approach to identify the relevant edges, and finally we contract all the edges that are not relevant. 

	Conceptually, to execute a single contraction process, we execute a $2$-out contraction, after which we contract e set of edges $E_p$, to which we include the edges of the input graph with probability $p = \frac{1}{2\delta}$. By \cref{l:node_compression_number_of_components} we know that after $2$-out contraction we have a graph that has only $\bigO\paren{\frac{n}{\delta}}$ vertices. By \cref{lem:random_edge_reduction}, contracting $E_{1/(2\delta)}$ reduces the number of edges further down to $\bigO(n)$ while preserving a fixed minimum cut with some constant probability. Therefore, combining $2$-out contraction with contracting the edges of $E_{1/(2\delta)}$ gives a contraction process that meets the guarantees from \cref{t:main_theorem}.

In order to give an efficient implementation, we execute to steps of the reduction simultaneously, i.e. we contract all the edges from $2$-out subgraph and uniformly sampled edges, i.e. we contract connected components of a graph $(V, I_1 \cup I_2 \cup E_{1/(2\delta)})$, and then identify the edges between the connected components.

For a single contraction process the total number of edges is $\bigO(\frac{m}{\delta} + n)$, hence straightforward application of the Connected Component algorithm allows to identify the connected components in $\bigO(1)$ rounds ~\cite{JurdzinskiNowicki2018}, with $\bigO(n)$ memory limit on a single machine and $\bigO\paren{\frac{m}{\delta} + n \log^2 n}$ global memory \footnote{The paper ~\cite{JurdzinskiNowicki2018} gives an $\bigO(1)$ algorithm for the MST problem in the Congested Clique model, which can be simulated in the \MPC model with a $\bigO(n)$ memory limit of a single machine. Furthermore, small changes in the analysis provided in \cite{JurdzinskiNowicki2018} give $\bigO(m + n\log^2 n)$ bound on global memory of the Connected Components algorithm.}. To complete the contraction process, it is enough to check for each edge, whether it is a inter component edge or not, which can be done by comparing $(V, I_1 \cup I_2 \cup E_{1/(2\delta)})$ component id-s for of the endpoints of the edge, which can be done in $\bigO(1)$ rounds, with $\bigO(n)$ limit on local memory of a single machine with $\bigO(m)$ global memory, e.g. via sorting algorithm.

Therefore, executing $\Theta(\log n)$ contraction processes in parallel can be done in $\bigO(1)$ rounds, with $\bigO(n)$ limit on a memory of a single machine and $\bigO(m \log n + n \log^3 n)$ memory. Furthermore the global memory bound can be improved to $\bigO(m + n \log^3 n)$ -- to do so, we use the fact that each contraction process starts from the same set of edges and the total number of inter component edges in all contraction processes is $\bigO(n \log n)$.

\subsection{\MPC algorithm with $\bigO(m + n\log^3 n)$ global memory}
The main ingredient of the algorithm is the protocol that given $k$ divisions into connected components allows to identify all inter component edges in $\bigO(m + n\cdot k + r \cdot k)$ global memory, where $r$ is the total number of inter component edges. This protocol allows us to execute $\Theta(\log n)$ connected computations in $\bigO(m + n \log^3 n)$ global memory and allows us to identify the edges of contracted graphs in $\bigO(m + n \log^2 n)$ memory. 

\paragraph{Computing the inter component edges in \MPC}
The naive approach would be to check for each edge, whether the endpoints are in the same connected component or not. This unfortunately, requires $\Theta(m)$ memory for a single contraction process, and $\Theta(m \log n)$ for $\Theta(\log n)$ contraction processes. To bypass this issue, we use the approach based on \emph{fingerprints}~\cite{Rabin1981FingerprintingBR}.
  
  The idea is roughly based on the fact that we can treat the labels of connected components of each vertex in $k$ contraction processes as $\Theta(k\log n)$ bit strings. Then, we can compute a hash value from some polynomial range, for each $\Theta(k\log n)$ bit label. Since we have only $\bigO(n)$ labels, if we use sufficiently large range of hashing function, with high probability there would be no collision. Since, the range is polynomial, the value of hash function can be encoded on $\bigO(1)$ words ($\bigO(\log n)$ bits). Then, for each edge, we compare fingerprints of the endpoints, if they are the same, the edge is not an inter component edge in any partition. Hence, in order to compute all edges that are inter-component, it is enough to consider only the edges with the endpoints with different fingerprints. For each of those edges we can simply gather component identifiers of their endpoints, which means that with $r$ inter-component edges we need to use only $\bigO(kr)$ global memory to do so.

\paragraph{Parallel connected component execution}
While the application to identifying the edges after contractions is straightforward, the application to parallel connected component computation may be not that clear, hence we briefly describe it. On the top level, we can use a KKT sampling approach with probability of sampling $\frac{1}{\log n}$ -- this gives us $\Theta(\log n)$ instances of the connected components problem, each with $\bigO\paren{\frac{m}{\log n}}$ edges. Running the Connected Components algorithm in parallel on those instances requires only $\bigO(m + n \log^3 n)$ global memory. Furthermore, the total number of inter component edges for each instance is $\bigO(n \log n)$, with high probability. Therefore, we can use the fingerprint based approach to identify all of them, using $\bigO(m + n \cdot \log n + n\log^2 n \cdot \log n)$ global memory. The resulting instances have only $\bigO(n \log n)$ edges, hence we can solve all of them in parallel in $\bigO(\log n \cdot (n \log n + n \log^2 n)) = \bigO(n \log^3 n)$ global memory. Therefore, total memory requirement of this part is $\bigO(m + n\log^3 n)$.
	
\section{\MC in the PRAM model} \label{s:mc_pram}
  In this section we provide a \CPRAM algorithm for the \MC problem for simple graphs, which is a proof of \Cref{t:mc_pram}. 
  
  The \CPRAM model is a model of parallel computing. The \PRAM machine consists of a set of $p$ processors, and some unbounded shared memory. The computations are performed in synchronous steps, and in each step each processor may read from $\bigO(1)$ memory cells, evaluate some $\bigO(1)$ step computable function on read values, and write something to $\bigO(1)$ memory cells. More precisely, we consider \CPRAM model, which extends to \emph{Concurrent-Read-Exclusive-Write \PRAM}, which means that we allow multiple processors to read from the same memory cell, but we forbid multiple processors to write to a single memory cell in a single step of computation.
  
  To define the complexity of an algorithm in the \CPRAM model, one can use the \emph{Work-Depth} model~\cite{Blelloch_1996_work_depth}. In this model, we perceive a computation as a directed acyclic graph, in which each vertex corresponds to a single step of a processor, its in-edges correspond to the inputs of evaluated function, and out-edges correspond to the results of the evaluated function. In other words, we put an edge between two vertices, if the output of the function evaluated by one vertex is an input of the function evaluated in the other vertex. The work of the algorithm is the number of vertices in the graph, and the depth of the algorithm is the longest directed path in the graph of computation.
  
  The state of the art algorithm for the weighted \MC problem has $\bigO(m\log^4 n)$ work and $\bigO(\log^3 n)$ depth~\cite{Geissmann2018_PRAM_MC}. The application of \Cref{t:main_theorem} allows us to reduce the number of edges to $\bigO(n)$. In this section we show an \PRAM implementation of \Cref{t:main_theorem} that allows us to execute $\Theta(\log n)$ contraction processes in $\bigO(m \log n)$ work and $\bigO(\log n)$ depth. The next step is to apply the technique from \Cref{subsec:succAmp}, which can be done in $\bigO(m + n\log n)$ work with $\bigO(\log n)$ depth, and gives a $\bigO(n)$ edge multigraph preserving $(2-\eps)$-minimum cuts with high probability, as stated in \Cref{thm:main}. To complete \MC computation, we still have to identify \MC in the resulting multigraph, for which we can use the state of the art algorithm for general graphs~\cite{Geissmann2018_PRAM_MC}, which has $\bigO(\log^3 n)$ depth, but only $\tilde\bigO(n)$ work.
  
  \begin{lemma}\label{l:mc_pram}
    Given a simple input graph $G$, with $n$ vertices, $m$ edges, and minimum degree $\delta$ it is possible to compute a $\bigO(n)$ edge multigraph, that preserves all non singleton $(2-\eps)$-minimum cuts with high probability on a \CPRAM machine, with $\bigO(m \log n)$ work and depth $\bigO(\log n)$. With high probability one of computed contractions preserves a fixed non singleton minimum cut of $G$.
  \end{lemma}
  In the remaining part this section, we prove \Cref{l:mc_pram}, and briefly discuss that composition of \Cref{l:mc_pram} with the state of the art algorithm for general graphs~\cite{Geissmann2018_PRAM_MC} proves \Cref{t:mc_pram}.
	
\paragraph{Contraction process in \CPRAM}
Similarly as for the \MPC model, the contraction process we implement consists of the $2$-out contraction [\cref{l:node_compression_number_of_components}] and contraction of uniformly sampled edges [\cref{lem:random_edge_reduction}]. Therefore, a single contraction process is basically a connected component computation on a graph with $\bigO(m)$ edges, which can be done with $\bigO(\log n)$ depth and $\bigO(m)$ work. In order to identify the inter component edges, we simply compare the identifiers of the connected components of the endpoints of the edge, hence this can be done in $\bigO(m)$ work and $\bigO(1)$ steps. Therefore, executing $\Theta(\log n)$ contraction processes can be done with $\bigO(m \log n)$ work and $\bigO(\log n)$ depth.

  \paragraph{Merging the results of contractions}
  At this point, we have $\Theta(\log n)$ multigraphs and each has only $\bigO(n)$ edges. By using the probability amplifying technique described in \Cref{subsec:succAmp}, we can transform them into a single graph with $\bigO(n)$ edges that preserves a all non singleton $(2-\epsilon)$ minimum cuts. To do so, it is enough to compute for each edge what is the number of contraction processes, which preserved that edge and keep only those that were preserved $r \in \Omega(\log n)$ times, and contract all other edges. More precisely, if $E_r$ is the set of edges we want to preserve, we contract all connected components of the graph $G'=(V, E \setminus E_r)$. 
  
  In order to compute the number of contraction processes that preserved an edge, we can simply scan over all the results of contraction processes, which requires $\bigO(n \log n)$ work, and has depth $\bigO(\log n)$, and split the set of edges into $E \setminus E_r$ and $E_r$, which can be done via parallel prefix computation in $\bigO(m)$ work, with depth $\bigO(\log n)$. The last part is a connected component computation, which can be done in $\bigO(m)$ work and $\bigO(\log n)$ depth, and relabeling the edges of $E_r$ so that for each edge we could know what are the identifiers of the endpoints after contractions, which also can be done in $\bigO(m)$ work and $\bigO(\log n)$. Therefore, merging the results of contractions from \cref{t:main_theorem} into a graph from \cref{thm:main} requires $\bigO(m + n\log n)$ work and $\bigO(\log n)$ depth. This concludes the proof of \Cref{l:mc_pram}.
  
  \paragraph{Computing \MC}
  At this point, we have a single multigraph with $\bigO(n)$ edges that preserves a minimum cut with high probability. The state of the art \CPRAM algorithm can compute its minimum cut in $\bigO(n \log^4 n)$ work, with depth $\bigO(\log^3 n)$.
  
  Therefore, the whole algorithm requires $\bigO(m \log n + n \log^4 n)$ work and has $\bigO(\log n + \log^3 n) =  \bigO(\log^3 n)$ depth, which concludes the proof of \Cref{t:mc_pram}.

	\newpage
	\subsection*{Acknowledgment} The first two authors are thankful to the Uber driver in Wroclaw whose late arrival provided ample time for a conversation about (massively parallel) algorithms for min-cut; it was during that conversation that the idea of random out contractions was sparked. We are also thankful to Danupon Nanongkai for discussions about the work in Daga et al.~\cite{daga2019} and for informing us that if we can upper bound the diameter of the components in $2$-out, we can further improve our distributed algorithms using the algorithm of Daga et al.~\cite{daga2019}; that led us to prove \Cref{lem:diameter}, which improved our distributed round complexity from $\tilde{O}(n^{1-1/9} D^{1/9} + n^{1-1/18})$ to $\tilde{O}(n^{0.8} D^{0.2} + n^{0.9})$.
  
\bibliographystyle{alpha} 
\bibliography{ref}

\newcommand{\etalchar}[1]{$^{#1}$}
\begin{thebibliography}{GGK{\etalchar{+}}18}

\bibitem[ABB{\etalchar{+}}19]{assadi2019coresets}
Sepehr Assadi, MohammadHossein Bateni, Aaron Bernstein, Vahab Mirrokni, and
  Cliff Stein.
\newblock Coresets meet edcs: algorithms for matching and vertex cover on
  massive graphs.
\newblock In {\em Pro.\ of ACM-SIAM Symp.\ on Disc.\ Alg.\ (SODA)}, 2019.

\bibitem[ASS{\etalchar{+}}18]{andoni2018parallel}
Alexandr Andoni, Clifford Stein, Zhao Song, Zhengyu Wang, and Peilin Zhong.
\newblock Parallel graph connectivity in log diameter rounds.
\newblock In {\em Proc.\ of the Symp.\ on Found.\ of Comp.\ Sci.\ (FOCS)},
  pages 674--685, 2018.

\bibitem[ASW19]{assadi2019massively}
Sepehr Assadi, Xiaorui Sun, and Omri Weinstein.
\newblock Massively parallel algorithms for finding well-connected components
  in sparse graphs.
\newblock In {\em the Proc.\ of the Int'l Symp.\ on Princ.\ of Dist.\ Comp.\
  (PODC)}, page to appear, 2019.

\bibitem[BBD{\etalchar{+}}19]{behnezhad2019TreeMIS}
Soheil Behnezhad, Sebastian Brandt, Masha Derakhshan, Manuela Fischer,
  MohammadTaghi Hajiaghayi, Richard~M. Karp, and Jara Uitto.
\newblock Massively parallel computation of matching and mis in sparse graphs.
\newblock In {\em the Proc.\ of the Int'l Symp.\ on Princ.\ of Dist.\ Comp.\
  (PODC)}, page to appear, 2019.

\bibitem[BEG{\etalchar{+}}18]{boroujeni2018approximating}
Mahdi Boroujeni, Soheil Ehsani, Mohammad Ghodsi, MohammadTaghi HajiAghayi, and
  Saeed Seddighin.
\newblock Approximating edit distance in truly subquadratic time: quantum and
  mapreduce.
\newblock In {\em Pro.\ of ACM-SIAM Symp.\ on Disc.\ Alg.\ (SODA)}, pages
  1170--1189, 2018.

\bibitem[BFU19]{brandt2019breaking}
Sebastian Brandt, Manuela Fischer, and Jara Uitto.
\newblock Breaking the linear-memory barrier in mpc: Fast mis on trees with
  strongly sublinear memory.
\newblock In {\em 26th International Colloquium on Structural Information and
  Communication Complexity}, page to appear, 2019.

\bibitem[BHH19]{behnezhad2019exponentially}
Soheil Behnezhad, MohammadTaghi Hajiaghayi, and David~G Harris.
\newblock Exponentially faster massively parallel maximal matching.
\newblock In {\em Proc.\ of the Symp.\ on Found.\ of Comp.\ Sci.\ (FOCS)}, page
  to appear, 2019.

\bibitem[Ble96]{Blelloch_1996_work_depth}
Guy~E. Blelloch.
\newblock Programming parallel algorithms.
\newblock {\em Commun. ACM}, 39(3):85--97, March 1996.

\bibitem[CFG{\etalchar{+}}19]{chang2019coloring}
Yi-Jun Chang, Manuela Fischer, Mohsen Ghaffari, Jara Uitto, and Yufan Zheng.
\newblock The complexity of (delta + 1)-coloring in congested clique, massively
  parallel computation, and centralized local computation.
\newblock In {\em the Proc.\ of the Int'l Symp.\ on Princ.\ of Dist.\ Comp.\
  (PODC)}, page to appear, 2019.

\bibitem[CLM{\etalchar{+}}18]{czumaj2017round}
Artur Czumaj, Jakub Lacki, Aleksander Madry, Slobodan Mitrovic, Krzysztof Onak,
  and Piotr Sankowski.
\newblock Round compression for parallel matching algorithms.
\newblock In {\em Proc.\ of the Symp.\ on Theory of Comp.\ (STOC)}, pages
  471--484, 2018.

\bibitem[DG04]{dg04}
Jeffrey Dean and Sanjay Ghemawat.
\newblock {MapReduce}: Simplified data processing on large clusters.
\newblock In {\em Proceedings of the 6th Conference on Symposium on Operating
  Systems Design \& Implementation (OSDI)}, pages 10--10, Berkeley, CA, USA,
  2004. USENIX Association.

\bibitem[DHNS19]{daga2019}
Mohit Daga, Monika Henzinger, Danupon Nanongkai, and Thatchaphol Saranurak.
\newblock Distributed edge connectivity in sublinear time.
\newblock In {\em Proceedings of the twenty-third annual ACM symposium on
  Theory of computing}, page to appear. ACM, 2019.

\bibitem[DKL76]{DKL76}
Efim~A. Dinitz, A.~V. Karzanov, and Micael~V. Lomonosov.
\newblock On the structure of a family of minimum weighted cuts in a graph.
\newblock In A.~A. Fridman, editor, {\em Studies in Discrete Optimization},
  pages 290--306. Nauka, Moskow, 1976.
\newblock (in Russian).

\bibitem[Doe18]{DBLP:journals/corr/abs-1801-06733}
Benjamin Doerr.
\newblock Probabilistic tools for the analysis of randomized optimization
  heuristics.
\newblock {\em CoRR}, abs/1801.06733, 2018.

\bibitem[FF56]{Ford-Fulkerson_algo}
L.~R. Ford and D.~R. Fulkerson.
\newblock Maximal flow through a network.
\newblock {\em Canadian Journal of Mathematics}, 8:399--404, 1956.

\bibitem[FF62]{ford1962flows}
Lestor~R Ford and DR~Fulkerson.
\newblock Flows in networks.
\newblock 1962.

\bibitem[FJ17]{frieze2017random}
Alan Frieze and Tony Johansson.
\newblock On random k-out subgraphs of large graphs.
\newblock {\em Random Structures \& Algorithms}, 50(2):143--157, 2017.

\bibitem[Gab91]{Gabow1991}
Harold~N. Gabow.
\newblock A matroid approach to finding edge connectivity and packing
  arborescences.
\newblock In {\em Proc.\ of the Symp.\ on Theory of Comp.\ (STOC)}, pages
  112--122. ACM, 1991.

\bibitem[Gab16]{Gab16}
Harold~N. Gabow.
\newblock The minset-poset approach to representations of graph connectivity.
\newblock {\em {ACM} Trans. Algorithms}, 12(2):24:1--24:73, 2016.
\newblock Announced at FOCS'91.

\bibitem[GG18]{Geissmann2018_PRAM_MC}
Barbara Geissmann and Lukas Gianinazzi.
\newblock Parallel minimum cuts in near-linear work and low depth.
\newblock In {\em Proceedings of the 30th on Symposium on Parallelism in
  Algorithms and Architectures}, SPAA '18, pages 1--11, New York, NY, USA,
  2018. ACM.

\bibitem[GGK{\etalchar{+}}18]{ghaffari2018improved}
Mohsen Ghaffari, Themis Gouleakis, Christian Konrad, Slobodan Mitrovi{\'c}, and
  Ronitt Rubinfeld.
\newblock Improved massively parallel computation algorithms for mis, matching,
  and vertex cover.
\newblock In {\em the Proc.\ of the Int'l Symp.\ on Princ.\ of Dist.\ Comp.\
  (PODC)}. arXiv:1802.08237, 2018.

\bibitem[GH61]{gomory1961multi}
Ralph~E Gomory and Tien~Chung Hu.
\newblock Multi-terminal network flows.
\newblock {\em Journal of the Society for Industrial and Applied Mathematics},
  9(4):551--570, 1961.

\bibitem[GHT18]{GHT18:inc-edgeconn}
Gramoz Goranci, Monika Henzinger, and Mikkel Thorup.
\newblock Incremental exact min-cut in polylogarithmic amortized update time.
\newblock {\em {ACM} Trans. Algorithms}, 14(2):17:1--17:21, 2018.

\bibitem[GKMS19]{gamlath2018weighted}
Buddhima Gamlath, Sagar Kale, Slobodan Mitrovi{\'c}, and Ola Svensson.
\newblock Weighted matchings via unweighted augmentations.
\newblock In {\em the Proc.\ of the Int'l Symp.\ on Princ.\ of Dist.\ Comp.\
  (PODC)}, page to appear, 2019.

\bibitem[GKU19]{ghaffari2019conditionalLB}
Mohsen Ghaffari, Fabian Kuhn, and Jara Uitto.
\newblock Conditional hardness results for massively parallel computation from
  distributed lower bounds.
\newblock In {\em Proc.\ of the Symp.\ on Found.\ of Comp.\ Sci.\ (FOCS)}, page
  to appear, 2019.

\bibitem[GU19]{ghaffari2019sparsifying}
Mohsen Ghaffari and Jara Uitto.
\newblock Sparsifying distributed algorithms with ramifications in massively
  parallel computation and centralized local computation.
\newblock In {\em Pro.\ of ACM-SIAM Symp.\ on Disc.\ Alg.\ (SODA)}, pages
  1636--1653, 2019.

\bibitem[HKT{\etalchar{+}}19]{HKTZZ19:k-out}
Jacob Holm, Valerie King, Mikkel Thorup, Or~Zamir, and Uri Zwick.
\newblock Random $k$-out subgraph leaves only ${O}(n/k)$ inter-component edges,
  2019.
\newblock To appear at FOCS'19.

\bibitem[HRW17]{DBLP:conf/soda/HenzingerRW17}
Monika Henzinger, Satish Rao, and Di~Wang.
\newblock Local flow partitioning for faster edge connectivity.
\newblock In {\em Proceedings of the Twenty-Eighth Annual {ACM-SIAM} Symposium
  on Discrete Algorithms, {SODA} 2017, Barcelona, Spain, Hotel Porta Fira,
  January 16-19}, pages 1919--1938, 2017.

\bibitem[IBY{\etalchar{+}}07]{Isard:2007}
Michael Isard, Mihai Budiu, Yuan Yu, Andrew Birrell, and Dennis Fetterly.
\newblock Dryad: Distributed data-parallel programs from sequential building
  blocks.
\newblock {\em SIGOPS Operating Systems Review}, 41(3):59--72, 2007.

\bibitem[JN18]{JurdzinskiNowicki2018}
Tomasz Jurdzi\'{n}ski and Krzysztof Nowicki.
\newblock {MST in {O}(1) Rounds of Congested Clique}.
\newblock In {\em Pro.\ of ACM-SIAM Symp.\ on Disc.\ Alg.\ (SODA)}, pages
  2620--2632, 2018.

\bibitem[Kar93]{Karger1993ContracitonAlgorithm}
David Karger.
\newblock Global min-cuts in ${R}{N}{C}$ and other ramifications of a simple
  mincut algorithm.
\newblock In {\em Pro.\ of ACM-SIAM Symp.\ on Disc.\ Alg.\ (SODA)}, pages
  21--30, 01 1993.

\bibitem[Kar96]{DBLP:conf/stoc/Karger96}
David~R. Karger.
\newblock Minimum cuts in near-linear time.
\newblock In {\em Proceedings of the Twenty-Eighth Annual {ACM} Symposium on
  the Theory of Computing, Philadelphia, Pennsylvania, USA, May 22-24, 1996},
  pages 56--63, 1996.

\bibitem[Kar00]{Karger2000}
David~R. Karger.
\newblock Minimum cuts in near-linear time.
\newblock {\em J. ACM}, 47(1):46--76, January 2000.

\bibitem[KKT95]{Karger:1995:RLA:201019.201022}
David~R. Karger, Philip~N. Klein, and Robert~E. Tarjan.
\newblock A randomized linear-time algorithm to find minimum spanning trees.
\newblock {\em J. ACM}, 42(2):321--328, March 1995.

\bibitem[KP09]{DBLP:conf/soda/KargerP09}
David~R. Karger and Debmalya Panigrahi.
\newblock A near-linear time algorithm for constructing a cactus representation
  of minimum cuts.
\newblock In {\em Proc. 20th SODA}, pages 246--255, 2009.

\bibitem[KS93]{DBLP:conf/stoc/KargerS93}
David~R. Karger and Clifford Stein.
\newblock An $\tilde{O}(n^2)$ algorithm for minimum cuts.
\newblock In {\em Proceedings of the Twenty-Fifth Annual {ACM} Symposium on
  Theory of Computing, May 16-18, 1993, San Diego, CA, {USA}}, pages 757--765,
  1993.

\bibitem[KSV10]{KarloffSV10}
Howard~J. Karloff, Siddharth Suri, and Sergei Vassilvitskii.
\newblock A model of computation for {MapReduce}.
\newblock In {\em Pro.\ of ACM-SIAM Symp.\ on Disc.\ Alg.\ (SODA)}, pages
  938--948, 2010.

\bibitem[KT15]{DBLP:conf/stoc/KawarabayashiT15}
Ken{-}ichi Kawarabayashi and Mikkel Thorup.
\newblock Deterministic global minimum cut of a simple graph in near-linear
  time.
\newblock In {\em Proceedings of the Forty-Seventh Annual {ACM} on Symposium on
  Theory of Computing, {STOC} 2015, Portland, OR, USA, June 14-17, 2015}, pages
  665--674, 2015.

\bibitem[KT19]{KT19:edge-conn}
Ken{-}ichi Kawarabayashi and Mikkel Thorup.
\newblock Deterministic edge connectivity in near-linear time.
\newblock {\em J. {ACM}}, 66(1):4:1--4:50, 2019.

\bibitem[LMSV11]{lattanzi2011filtering}
Silvio Lattanzi, Benjamin Moseley, Siddharth Suri, and Sergei Vassilvitskii.
\newblock Filtering: a method for solving graph problems in mapreduce.
\newblock In {\em the Proceedings of the Symposium on Parallel Algorithms and
  Architectures}, pages 85--94, 2011.

\bibitem[MK13]{distributed-cut}
{M. Ghaffari} and Fabian Kuhn.
\newblock Distributed minimum cut approximation.
\newblock In {\em Proc.\ of the Int'l Symp.\ on Dist.\ Comp.\ (DISC)}, pages
  1--15, 2013.

\bibitem[MR95]{MR95}
Rajeev Motwani and Prabhakar Raghavan.
\newblock {\em Randomized Algorithms}.
\newblock Cambridge University Press, 1995.

\bibitem[NI92]{IbarakiNagamochi1992}
Hiroshi Nagamochi and Toshihide Ibaraki.
\newblock Computing edge-connectivity in multigraphs and capacitated graphs.
\newblock {\em SIAM Journal on Discrete Mathematics}, 5(1):54--66, 1992.

\bibitem[NS14]{Nanongkai2014-distributed-cut}
Danupon Nanongkai and Hsin-Hao Su.
\newblock Almost-tight distributed minimum cut algorithms.
\newblock In {\em Proc.\ of the Int'l Symp.\ on Dist.\ Comp.\ (DISC)}, pages
  439--453, 2014.

\bibitem[NW61]{Nash-Williams}
C.~St. J.~A. Nash-Williams.
\newblock Edge-disjoint spanning trees of finite graphs.
\newblock {\em J.\ of the London Math.\ Society}, 36:445--450, 1961.

\bibitem[Pel00]{Peleg:2000}
David Peleg.
\newblock {\em Distributed Computing: A Locality-sensitive Approach}.
\newblock Society for Industrial and Applied Mathematics, Philadelphia, PA,
  USA, 2000.

\bibitem[PR99]{Pettie2004_MSF_EREW_PRAM}
Seth Pettie and Vijaya Ramachandran.
\newblock A randomized time-work optimal parallel algorithm for finding a
  minimum spanning forest.
\newblock In Dorit~S. Hochbaum, Klaus Jansen, Jos{\'e} D.~P. Rolim, and
  Alistair Sinclair, editors, {\em Randomization, Approximation, and
  Combinatorial Optimization. Algorithms and Techniques}, pages 233--244,
  Berlin, Heidelberg, 1999. Springer Berlin Heidelberg.

\bibitem[Rab81]{Rabin1981FingerprintingBR}
Michael~O. Rabin.
\newblock Fingerprinting by random polynomials.
\newblock 1981.

\bibitem[Tar75]{T75}
R.~E. Tarjan.
\newblock Efficiency of a good but not linear set union algorithms.
\newblock {\em J. ACM}, 22:215--225, 1975.

\bibitem[Tut61]{Tutte}
W.~T. Tutte.
\newblock On the problem of decomposing a graph into {$n$} connected factors.
\newblock {\em J.\ of the London Math.\ Society}, 36:221--230, 1961.

\bibitem[Whi12]{White:2012}
Tom White.
\newblock {\em Hadoop: The Definitive Guide}.
\newblock O'Reilly Media, Inc., 2012.

\bibitem[WM93]{Matula93}
David W.~Matula.
\newblock A linear time 2+epsilon approximation algorithm for edge
  connectivity.
\newblock pages 500--504, 01 1993.

\bibitem[ZCF{\etalchar{+}}10]{ZahariaCFSS10}
Matei Zaharia, Mosharaf Chowdhury, Michael~J. Franklin, Scott Shenker, and Ion
  Stoica.
\newblock Spark: Cluster computing with working sets.
\newblock In {\em 2nd {USENIX} Workshop on Hot Topics in Cloud Computing
  (HotCloud)}, 2010.

\end{thebibliography}

\end{document}